\documentclass[onecolumn]{IEEEtran}
\usepackage{calc}
\usepackage[T1]{fontenc}

\usepackage[overload]{empheq}
\usepackage{cleveref}
\usepackage{multirow}
\usepackage{cite}
\usepackage{graphicx,subfigure}
\usepackage{psfrag}
\usepackage{amsmath,amssymb}
\usepackage{color}
\usepackage{tikz}
\usepackage{cleveref}
\usepackage{algorithm}
\usepackage{algorithmic}
\usepackage{empheq}
\usepackage{setspace}
\newenvironment{proof}{\begin{IEEEproof}}{\end{IEEEproof}}
\newcommand*{\QEDA}{\hfill\ensuremath{\blacksquare}}

\DeclareMathOperator*{\defeq}{\triangleq}

\newtheorem{lemma}{Lemma}

\newtheorem{proposition}{Proposition}
 
\newcommand{\bit}{\begin{itemize}}
\newcommand{\eit}{\end{itemize}}

\newcommand{\bc}{\begin{center}}
\newcommand{\ec}{\end{center}}

\newcommand{\ba}{\begin{array}}
\newcommand{\ea}{\end{array}}

\newcommand{\beq}{\begin{equation}}
\newcommand{\eeq}{\end{equation}}

\newcommand{\beqn}{\begin{equation*}}
\newcommand{\eeqn}{\end{equation*}}

\newcommand{\bean}{\begin{eqnarray*}}
\newcommand{\eean}{\end{eqnarray*}}
\newcommand{\bea}{\begin{eqnarray}}
\newcommand{\eea}{\end{eqnarray}}

\def\C{\mathbb{C}}
\def\E{\mathbb{E}}

\def\dv{\boldsymbol{d}}

\def\hv{\boldsymbol{h}}

\def\vv{\boldsymbol{v}}

\def\xv{\boldsymbol{x}}

\newtheorem{remark}{Remark}

\makeatletter
\def\blfootnote{\gdef\@thefnmark{}\@footnotetext}
\makeatother

\begin{document}
\sloppy

\title{Fundamental Limits of Cloud and Cache-Aided Interference Management \\ with Multi-Antenna Edge Nodes}
\author{
   \IEEEauthorblockN{Jingjing Zhang and Osvaldo Simeone}
	}
\maketitle

\thispagestyle{empty}

%\begin{align} 
%F=F_T F_R = \frac{\text{l.c.m.} (m(\mu,r), K_T)}{m(\mu,r)} \frac{\text{l.c.m.} (u(\mu,r), K_R)}{K_R} 
%\end{align}
%
%
%\begin{align} 
%B = \frac{\text{l.c.m.} (m(\mu,r), K_T)}{m(\mu,r)} \frac{\text{l.c.m.} (u(\mu,r), K_R)}{u(\mu,r)} 
%\end{align}

\begin{abstract}
In fog-aided cellular systems, content delivery latency can be minimized by jointly optimizing edge caching and transmission strategies. \blfootnote{The authors are with the Department of Informatics, King's College London, London, UK (emails: jingjing.1.zhang@kcl.ac.uk, osvaldo.simeone@kcl.ac.uk).}In order to account for the cache capacity limitations at the Edge Nodes (ENs), transmission generally involves both fronthaul transfer from a cloud processor with access to the content library to the ENs, as well as wireless delivery from the ENs to the users. In this paper, the resulting problem is studied from an information-theoretic viewpoint by making the following practically relevant assumptions: 1) the ENs have multiple antennas; 2) only uncoded fractional caching is allowed; 3) the fronthaul links are used to send fractions of contents; and 4) the ENs are constrained to use one-shot linear precoding on the wireless channel. Assuming offline proactive caching and focusing on a high signal-to-noise ratio (SNR) latency metric, the optimal information-theoretic performance is investigated under both serial and pipelined fronthaul-edge transmission modes. The analysis characterizes the minimum high-SNR latency in terms of Normalized Delivery Time (NDT) for worst-case users' demands. The characterization is exact for a subset of system parameters, and is generally optimal within a multiplicative factor of $3/2$ for the serial case and of $2$ for the pipelined case. The results bring insights into the optimal interplay between edge and cloud processing in fog-aided wireless networks as a function of system resources, including the number of antennas at the ENs, the ENs' cache capacity and the fronthaul capacity.
\end{abstract}


\section{Introduction}
Content delivery is one of the most important use cases for mobile broadband services in 5G networks. A key technology that promises to help minimize delivery latency and network congestion is edge caching, which relies on the storage of popular contents at the ENs, i.e., at the base stations or access points. Initial works on the subject \cite{SGDFC:13} studied the advantages of edge caching in terms of cache hit probability, hence adopting the standard performance criteria used in the networking literature (see e.g., \cite{SMMPLS:13}). The information-theoretic analysis of edge caching, which has been undertaken in the past few years starting with \cite{MN:15isit}, has instead concentrated on the impact of the cached content distribution across the ENs on the ENs' capability to carry out interference management (see also \cite{LLC:14}).
As a general observation, caching the same content across multiple ENs enables cooperative delivery strategies involving multiple ENs, whereas properly placed distinct contents can yield coordination opportunities \cite{MN:15isit}. The relative effect of interference management via coordination or cooperation on content delivery is best studied in the high-SNR regime, in which the performance is limited by interference, as done in \cite{MN:15isit,NMA:17}.

In most existing works, as further discussed below, the high-SNR analysis of the interference management capabilities of cache-aided systems was performed under the assumption that the overall cache capacity available in the system, including at the users, is sufficient to store the entire library of popular contents. When this assumption is violated, contents need to be retrieved from a content server by leveraging transport links that connect the ENs to the access or core networks. This more general scenario was first studied from an information-theoretic perspective in \cite{TS:16,STS:17}. In these works, a cloud processor is assumed to be connected to the ENs via so called fronthaul links, as seen in Fig.~\ref{fig:model}.

For the model in Fig.~\ref{fig:model}, which is referred to as Fog-Radio Access Network (F-RAN), the key design problem concerns the optimal use of fronthaul and wireless edge resources for caching and delivery. Assuming the standard offline caching scenario with static popular set, reference \cite{STS:17} identified high-SNR optimal caching and delivery strategies within a multiplicative factor of two. 
The approximately optimal scheme in \cite{STS:17} relies on both Zero-Forcing (ZF) one-shot precoding and interference alignment for transmission on the wireless edge channel and on cloud precoding and on quantization \cite{QPSY:17}.
In this work, we revisit the results in \cite{STS:17} by making the following practically relevant assumptions: 1) the ENs have multiple antennas; 2) only uncoded fractional caching is allowed; 3) the fronthaul links can only be used to send uncoded fractions of contents; and 4) the ENs are constrained to use linear precoding on the wireless channel. 

\begin{figure}[t!] 
  \centering
\includegraphics[width=0.35\columnwidth]{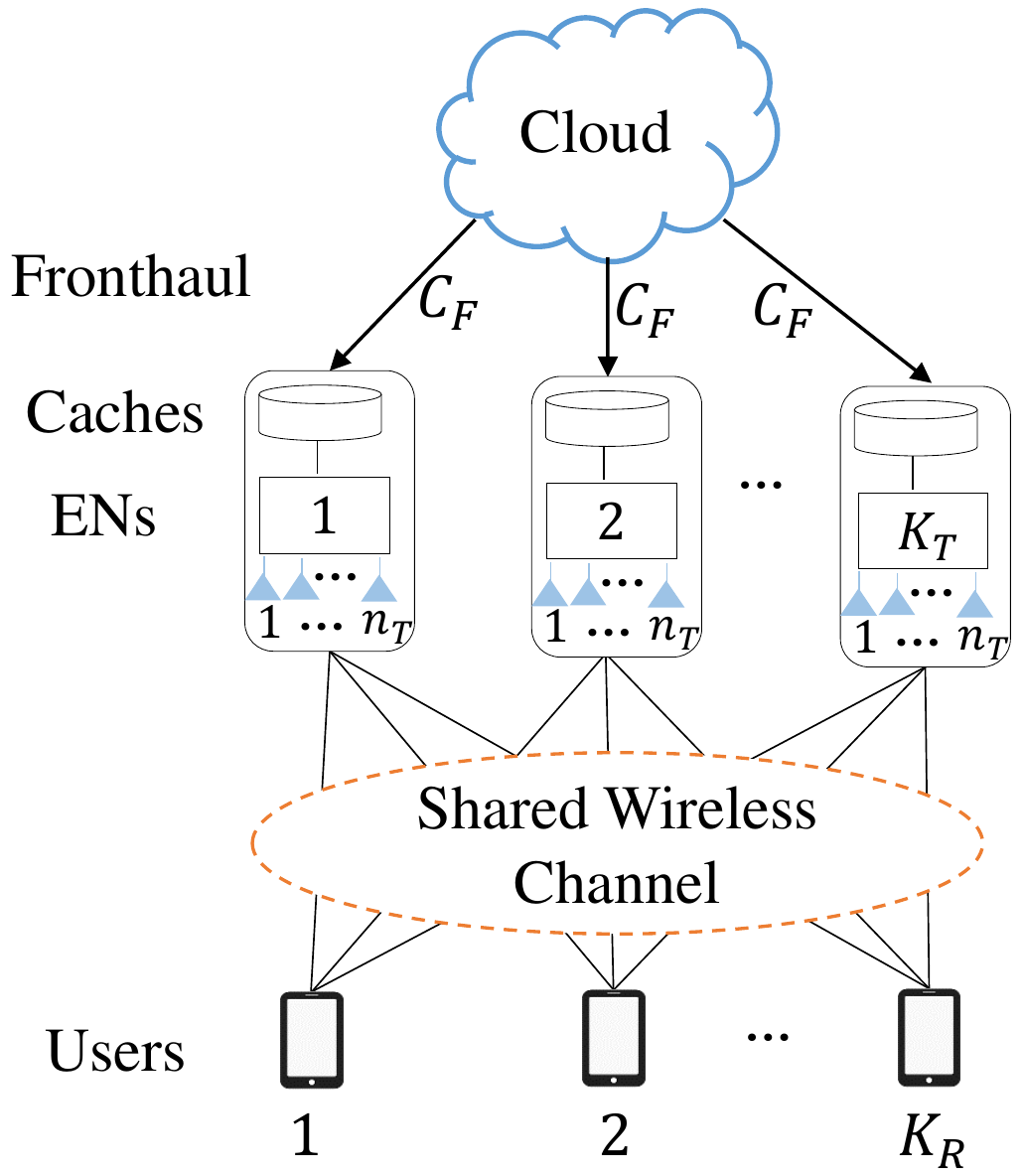}
\caption{Cloud and cache-aided F-RAN system model with multi-antenna ENs.}
\label{fig:model}
\end{figure}

\textbf{Related Work:}
Assuming offline caching, cache-aided interference management was first studied in \cite{MN:15isit}, in which transmitter-side caches are considered, and a delivery strategy is proposed by leveraging interference alignment, ZF precoding and interference cancellation. 
Extensions that account for caching at both transmitter and receiver sides were provided in \cite{CTXL:17,HND:16,RGT:17,NMA:17,NMA:17ICC}. In \cite{HND:16}, a novel strategy based on the separation of physical and network layers is investigated. Under the assumption of one-shot linear precoding, references \cite{NMA:17, NMA:17ICC} reveal that the transmitters' caches and receivers' caches contribute equally to the high-SNR performance. More general lower bounds were derived in \cite{STS:16CISS} (see also \cite{STS:17}), for which matching upper bounds were found in \cite{TS:16,RMTG:17} in special cases. Decentralized caching is studied in \cite{GENE:17}. In \cite{TG:16,MGL:17}, the performance of edge caching with partial connectivity is studied without any channel state information (CSI), while imperfect CSI is considered in \cite{KACSP:18,ZEinterplay:17}. 

The joint design of cloud processing and edge caching for the F-RAN model was studied in references \cite{TS:16,STS:17} and then in \cite{KGAS:17,GSP:17,SOAR:17}, by focusing on the high-SNR latency performance metric known as Normalized Delivery Time (NDT) proposed in \cite{STS:17}. This metric essentially measures the inverse of the number of degrees of freedom (DoF) \cite{TS:16,STS:17}. Reference \cite{STS:17} characterizes the minimum NDT with a multiplicative factor of 2 by considering single-antenna ENs, intra-file coding for caching and general delivery strategies on fronthaul and edge channels. References \cite{SMOR:17,KGAS:17} studied a related scenario with coexisting macro- and small-cell base stations. The work \cite{RTG:18} studied an F-RAN model with decentralized placement algorithms. In contrast to the abovementioned works, references \cite{SOAR:17,ASST:17} investigated online caching in the presence of a time-varying set of popular files.
Overall, the set-up studied in this work extends the model in \cite{NMA:17} by including cloud processing, fronthauling, and multi-antenna ENs, but, unlike \cite{NMA:17}, it excludes caching at the receivers' sides.

\textbf{Main contributions:}
This paper investigates interference management in a cloud and cache-aided F-RAN model, as illustrated in Fig.~1, with multiple antennas at the transmitters, under the assumptions of one-shot linear precoding and transmission of uncoded contents on the fronthaul links. The NDT is adopted as a high-SNR delivery latency metric. The main contributions are summarized as follows.
\begin{itemize}
  \item We first derive upper bounds on the minimum NDT under the assumption of serial fronthaul-edge transmission when only edge caching or only fronthaul resources (and no cache capacity) are available for delivery. In the serial delivery mode, fronthaul transmission is followed by edge transmission. The proposed schemes use clustered ENs' cooperation via ZF beamforming to cancel interference on the wireless edge channel. Cooperation is enabled by contents shared thanks to edge caching during placement phase, or by fronthaul transmissions in the delivery phase. The caching and fronthauling strategies rely on an efficient packetization method that is separately at most linear in the number of transmitters and receivers. 
	\item For the general F-RAN set-up, an upper bound on the minimum NDT is derived as a function of the cache storage capacity, the fronthaul rate, and the number of ENs' antennas. To this end, we propose a caching and delivery scheme that manages interference via ZF by means of the ENs' clustered cooperation as enabled by both fronthaul and edge caching resources. Then, an information-theoretic lower bound on the minimum NDT is derived. As a result, the minimum NDT is characterized exactly for a large subset of system parameters, and approximately within a multiplicative factor of $3/2$ for any value of the parameters. 
  \item We finally study a pipelined delivery mode whereby fronthaul and edge transmissions can take place simultaneously. We show that the NDT under pipelined transmission can be improved as compared to serial delivery, and the minimum NDT is derived within a multiplicative factor of 2.
\end{itemize}
The rest of the paper is organized as follows. Section~\ref{sec:model} describes a general $K_R\times K_T$ F-RAN model and the NDT performance metric for serial fronthaul-edge delivery. Section~\ref{sec:edge} and Section~\ref{sec:fronthaul} study the specific set-ups of edge-only and fronthaul-only F-RAN, respectively, while the general F-RAN set-up is investigated in Section~\ref{sec:edgeft}. Under serial transmission, the upper bounds and the proposed scheme are presented along with a finite optimality gap from the minimum NDT. The pipelined fronthaul-edge transmission is discussed in Section~\ref{sec:pipeline}, where upper and lower bounds on the minimum NDT are presented. Section~\ref{sec:con} concludes the work and also highlights future research directions. 

\textbf{Notation:}
For any integer $K$, we define the set $[K]\defeq\{1,2,\cdots,K\}$. For a set $A$, $|A|$ represents the cardinality. We use the notation $\{f_n\}_{n=1}^{N}\defeq\{f_1,\cdots, f_n, \cdots,f_N\}$. 
For function $g(n)$, the notation $f(n)=o(g(n))$ denotes a function $f(n)$ that satisfies the limit $\lim_{n\rightarrow \infty} (f(n)/g(n))=0$. The ceiling function $\lceil x \rceil $ maps $x$ to the least integer that is greater than or equal to $x$, and the floor function $\lfloor x \rfloor $ maps $x$ to the greatest integer that is less than or equal to $x$. Moreover, the nearest positive integer function $[x]$ returns the nearest positive integer to $x$. We also have $(x)^+\defeq \max\{x,0\}$.

\section{System Model and Performance Metric} \label{sec:model}
In this section, we present the model under study, which consists of an F-RAN system with multi-antenna ENs that performs the hard transfer of uncoded contents on the fronthaul links and one-shot linear precoding on the wireless edge channel. We also adapt the NDT metric \cite{STS:17} to this model. We consider the serial mode of delivering across fronthaul and edge channels, and discuss the pipelined mode in Section~\ref{sec:pipeline}.

\subsection{System Model} \label{model}
We consider the F-RAN model shown in Fig.~\ref{fig:model} where a set $\mathcal{K}_T=\{1,\cdots,K_T\}$ of ENs, each having $n_T$ antennas, are connected to $K_R$ single-antenna receivers $\mathcal{K}_R=\{1,\cdots,K_R\}$ through a shared wireless channel, as well as to a cloud processor (CP) via fronthaul links. The CP has access to a library of $N$ files $\{W_n\}_{n=1}^{N}$, of $L$ bits each. Any file $W_n$ contains $F$ packets $\mathcal{W}_n =\{W_{nf}\}_{f=1}^F$, where each packet $W_{nf}$ is of size $L/F$ bits, and $F$ is an arbitrary parameter. Note that we refer to the set of packets $\{W_{nf}\}_{f=1}^F$ in file $W_{n}$ as $\mathcal{W}_n$. Each fronthaul link has capacity $C_F$ bits per symbol, where a symbol refers to a channel use of the wireless channel, and each EN has a cache with capacity of $\mu NL$ bits, with $\mu \in [0,1]$. Parameter $\mu$ is referred to as the fractional cache size.

In the \emph{pre-fetching phase}, the caches of the ENs are pre-filled with content from the library under the cache capacity constraints. The content of the cache of each EN $i$ is described by the set $\mathcal{C}_i=\{\mathcal{C}_{i1}, \cdots,\mathcal{C}_{in}, \cdots,\mathcal{C}_{iN}\}$, where $\mathcal{C}_{in} \subseteq \mathcal{W}_n$ represents the subset of packets from file $W_n$ that are cached at EN $i$. Due to the cache capacity constraint, its size must satisfy the inequality 
\begin{align} \label{cap:cache}
\frac{|\mathcal{C}_{in}|}{F} \leq \mu.
\end{align} 
Note that as in~\cite{NMA:17}, the model at hand allows for no coding of the cached content either within or across files. 

In the \emph{delivery phase}, each user $k$ requests a file $W_{d_k}$, with $d_k \in [N]$, from the library. Given the request vector $\dv=\{d_1, \cdots, d_{K_R}\}$ and the CSI on the edge channel, to be discuss below, the CP transmits information about the requested files $\{W_{d_1},\cdots,W_{d_{K_R}}\}$ to the ENs via the fronthaul links. Specifically, on each fronthaul $i$, the set $\mathcal{F}_i=\{\mathcal{F}_{id_1}, \cdots, \mathcal{F}_{id_{K_R}}\}$ of packets is sent, where $\mathcal{F}_{id_k}\subseteq \mathcal{W}_{d_k}$ is a subset of packets from file $W_{d_k}$. Note that, as mentioned, the described model assumes hard-transfer fronthauling of uncoded packets. 
After the fronthaul transmission, any EN $i$ has access to the fronthaul information $\mathcal{F}_i$, as well as to the cached content $\mathcal{C}_i$. This information is used by the ENs to deliver the users' requests $\{W_{d_1},\cdots,W_{d_{K_R}}\}$  through the wireless channel. 
%To quantify the availability of the requested files at the ENs, we define the multiplicity $m(\mu,r)\leq K_T$ of any requested file as the number of times that the file appears across all the ENs after fronthaul transmission. The multiplicity hence accounts for both the pre-stored caching contents and the information received at the ENs from the fronthaul transmission. 

To this end, we constrain the wireless transmission strategy to one-shot linear precoding by following \cite{NMA:17}. Accordingly, wireless transmission takes place over $B$ blocks to deliver the $K_R F$ desired packets. In any block $b\in[B]$, the ENs send a subset of the requested packets, denoted by $\mathcal{D}(b) \subseteq \{W_{d_1f},\cdots, W_{d_Kf}\}_{f=1}^{F}$, to a subset $\mathcal{R}(b)$ of $K_T$ users, such that each user in $\mathcal{R}(b)$ can decode exactly one packet without interference by the end of the block. To this purpose, in any block $b$, each EN $i$ sends a linear combination of the subset of the packets in $\mathcal{D}(b)$ that it has access to in its cache, i.e., in $\mathcal{C}_i$, or that it has received on the fronthaul link, i.e., in $\mathcal{F}_i$. For any given symbol within the block, the transmitted signal of EN $i$ is hence given as 
\begin{align}
\xv_i(b) = \sum_{\substack {{(n,f):} \\ {W_{nf}\in  \mathcal{D}(b) \cap\{  \mathcal{C}_i\cup\mathcal{F}_i\}}}}  \vv_{inf}(b) s_{nf}(b),
\end{align}
where $s_{nf}(b)$ is a coded symbol for packet $W_{nf}$, and $\vv_{inf}(b) \in \C^{n_T\times 1}$ is the precoding vector for the same file. As we have described, each packet $W_{nf}\in  \mathcal{D}(b) \cap\{  \mathcal{C}_i\cup \mathcal{F}_i\}$ is intended for a single user in $\mathcal{R}(b)$. We impose the power constraint  $\E[||\xv_i(b)||^2] \leq P$.

The received signals of each user $k\in \mathcal{R}(b)$ in block $b$ is given as
\begin{align}
y_k(b)= \sum_{i=1}^{K_T} \hv_{ki}^T(b) \xv_i(b)+ z_k(b), 
\end{align}
where $\hv_{ki}(b) \in \C^{n_T\times 1}$ is the channel vector between EN $i$ and user $k$, and $z_k(b)$ is the zero-mean complex Gaussian noise with normalized unitary power. The channels $\{\hv_{ki}(b)\}_{k\in[K_R], i\in[K_T], b\in[B]}$ are arbitrary as long as the set $\{\hv_{ki}(b)\}_{k\in[K_R], i\in[K_T]}$ is linearly independent for each block $b$. In each block $b$, we assume that all the ENs and users have access to the full CSI $\{\hv_{ki}(b)\}_{k\in[K_T], i\in[K_R]}$ as necessary. The delivery of the packets in the set $\mathcal{D}(b)$ is said to be achievable if there exist precoding vectors $\{\vv_{inf}(b)\}$, such that, with full CSI, each user $k\in \mathcal{R}(b)$ can decode without interference its intended packet. Given that the users have a single antenna, this happens if the received signal $y_k(b)$ is directly proportional to the desired symbol $s_{nf}(b)$ plus additive Gaussian noise with constant power, i.e., not scaling with the signal power $P$. The resulting point-to-point interference-free channel from the ENs to each served user $k$ supports transmission at rate $\log(P)+o(\log(P))$.

\subsection{Performance Metric: NDT}
Consider a given policy defined by the parameters $\{\mathcal{C}_i, \mathcal{F}_i, \{\vv_{inf}(b)\}_{n\in[N], f\in[F],b\in[B]}\}_{i=1}^{K_T}$, where the fronthaul messages $\{\mathcal{F}_i\}_{i=1}^{K_T}$ and the beamforming vectors $\{\vv_{inf}(b)\}_{n\in[N], f\in[F],b\in[B], i\in[K_T]}$ are defined on request vector $\dv$ and CSI $\{\hv_{ki}(b)\}_{k\in[K_R], i\in[K_T], b\in[B]}$. Given the fronthaul messages defined by the subsets $\{\mathcal{F}_i\}_{i=1}^{K_T}$, the time required for fronthaul transmission can be computed as  
\begin{align}
T_F = \max_{i\in[K_T]}\frac{|\mathcal{F}_i|L}{F}\frac{1}{C_F}, \label{def:TF}
\end{align}
since $|\mathcal{F}_i|$ packets with $|\mathcal{F}_i|L/F$ bits need to be delivered to EN $i$ over a fronthaul link of capacity $C_F$ and $T_F$ is the maximum among the $K_T$ fronthaul latencies. Furthermore, given the delivered packet set $\{\mathcal{D}(b)\}_{b=1}^{B}$, the total time needed for wireless edge transmission over $B$ blocks is  
\begin{align}
T_E=\frac{BL}{F}\frac{1}{(\log(P) + o(\log(P))}. \label{def:TE}
\end{align}
This is because, in each of the $B$ blocks, one packet with $L/F$ bits is sent to each user in $\mathcal{R}(b)$ at rate $\log(P) + o(\log(P))$.

As in~\cite{STS:17}, we normalize the latency by the term $L/\log(P)$. This corresponds to the transmission latency, neglecting  $o(\log(P))$ terms, for a reference system that transmits interference-free to all users at the maximum rate $\log(P)$. Moreover, as in~\cite{STS:17}, we evaluate the impact of the fronthaul capacity $C_F$ in the high-SNR regime by using the scaling $C_F = r\log(P)$, so that the parameter $r$ measures the ratio between the fronthaul capacity and the interference-free wireless channel capacity to any user. Accordingly, we define the fronthaul NDT of the given policy as 
\begin{align}
\delta_F = \lim_{P\rightarrow\infty} \lim_{L\rightarrow\infty} \frac{T_F}{L/\log(P)}= \max_{i\in[K_T]} \frac{|\mathcal{F}_i|}{F r}, \label{def:DF}
\end{align}
and the edge NDT as
\begin{align}
\delta_E =\lim_{P\rightarrow\infty} \lim_{L\rightarrow\infty}  \frac{T_E}{L/\log(P)} =  \frac{B}{F}. \label{def:DE}
\end{align}
Assuming serial fronthaul and edge transmission, the overall NDT is given as 
\begin{align} \label{def:deltaach}
\delta = \delta_E+\delta_E.
\end{align}
For any pair $(\mu,r)$, the minimal NDT across all achievable policies $\{\mathcal{C}_i, \mathcal{F}_i, \{\vv_{inf}(b)\}_{n\in[N], f\in[F],b\in[B]}\}_{i=1}^{K_T}$ is defined as 
\begin{align}
\bar{\delta}(\mu,r) = \inf\{\delta(\mu,r): \delta(\mu,r)~ \text{is achievable for some}~F \geq 1\}. \label{def:deltabar}
\end{align}
Note that in the definition \eqref{def:deltabar}, we allow for a partition of the files in an arbitrary number of $F$ packets. By construction, we have the inequality $\bar{\delta}(\mu,r) \geq 1$, where the lower bound is achieved in the mentioned ideal system. By allowing for time sharing among different policies, we finally define the minimum NDT as 
\begin{align}
\delta^*(\mu,r) = \text{l.c.e.}(\bar{\delta}(\mu,r)). \label{def:delta}
\end{align}
where the lower convex envelope (l.c.e.)\footnote{The l.c.e, is the supremum of all convex functions that lie under the given function.} is computed throughout this paper by considering $\bar{\delta}(\mu,r)$ as a function of $\mu$. The achievability of $\delta^*(\mu,r)$ given the achievable NDT $\bar{\delta}(\mu,r)$ follows by a standard cache and time-sharing argument, which is detailed in [7, Lemma 1].

\section{Achievable NDT for Edge-only Caching} \label{sec:edge}
As a preliminary result to be leveraged in Section~\ref{sec:edgeft}, here we describe an achievable NDT (Proposition 1) and the corresponding caching and delivery scheme for the described F-RAN model with edge-caching only, i.e., with zero fronthaul rate $(r=0)$. Note that in this regime, the condition $\mu \geq 1/K_T$ needs to be satisfied in order to ensure a finite NDT. In fact, otherwise, contents could not be fully cached across the ENs. The proposed scheme uses clustered ZF cooperation as in \cite{NMA:17} but via a more efficient packetization method.

\subsection{Achievable NDT} \label{subsec:NDTedge}
 
In the proposed scheme, in each block $b$, a cluster of ENs serve a given number $u$ of users by using cooperative ZF precoding on the wireless channel. Cooperation at the ENs via ZF is enabled by the availability of shared contents across the caches of the ENs. To quantify the content availability at the ENs, we define the multiplicity $m$ of any file as the number of times that the file appears across all the ENs. Via edge caching, a multiplicity $m(\mu)=\lfloor\mu K_T\rfloor$ can be ensured for all contents by edge caching since $\mu K_T$ is the per-file cache capacity across all the ENs. Given the multiplicity $m$, it will be shown that contents can be allocated so that clusters of $m$ ENs can transmit cooperatively in each block to serve up to $mn_T$ users on interference-free channels via ZF beamforming. Note that $mn_T$ is in fact the total number of transmit antennas available at a cluster of $m$ ENs. Hence, the number of users that can be served via ZF in each block is given as
\begin{align}\label{uedge}
u(m) = \min \{mn_T, K_R\}. 
\end{align}
Note that, by \eqref{uedge}, the multiplicity $m$ of a content can be upper bounded without loss of optimality by 
\begin{align} \label{mmax}
m_{max}= \min\left\{K_T,\left\lceil   \frac{K_R}{n_T} \right\rceil  \right\}.
\end{align}
This is because, when the multiplicity reaches $m_{max}$, the ENs can cooperate in each block via ZF beamforming to completely eliminate inter-user interference for the maximum number of users, which is given by $\min\{K_T n_T, K_R\}$. 
The resulting achievable NDT is presented in the following proposition. 
\begin{proposition} \label{pro:NDTedge}
For an F-RAN system with any cache capacity $\mu \in [1/K_T, 1]$ and fronthaul rate $r=0$, we have the upper bound on the minimum NDT $\delta^*(\mu,r=0)\leq \delta_{E}(m(\mu))$, where we have defined the edge NDT as a function of the multiplicity $m$ as
\begin{align} \label{ach:edge}
\delta_{E}(m) = \frac{K_R}{u(m)},
\end{align}
with function $u(m)$ in \eqref{uedge}, and the multiplicity
\begin{align} \label{mmu}
m(\mu)=\min\{ \lfloor \mu K_T \rfloor, m_{max}\}.
\end{align}
%\item for any $r\geq  r_{th}$,
%\begin{align}
%\delta_{ach}(\mu,r)=\frac{K_R (m_{max}-\mu K_T)^+}{ K_Tr} +\frac{K_R}{\min\{ K_Tn_T, K_R\}}. \label{NDT:bigr}
%\end{align}
\end{proposition}
\begin{proof}
The proof is reported in Section~\ref{sketchedge}.
\end{proof}

\subsection{Examples} \label{ex:edge}

Before proving a sketch of proof, we discuss two examples that illustrate the achievable cache-aided delivery scheme. We consider an F-RAN model with $r=0, K_T=4$ ENs and $n_T=2$ per-EN antennas (see Fig.~\ref{fig:schemee}). 

\emph{Example 1}. Consider $K_R=4$ and $\mu=0.5$, so that, by \eqref{mmu}, we have the multiplicity $m(\mu)=2$ in \eqref{mmu}. Each library file $W_n$ is divided into $F=2$ equal and disjoint packets $\{W_{n1}, W_{n2}\}$. As illustrated in Fig.~\ref{fig:schemee}(a), in the caching phase, the ENs are divided into two clusters of $m(\mu)=2$ ENs each: the first cluster, consisting of EN 1 and 2, caches the first packets $\{W_{n1}\}_{n=1}^{N}$, while the second cluster, consisting of EN 3 and 4, caches the second packets $\{W_{n2}\}_{n=1}^{N}$. As a result, we have the subsets $\mathcal{C}_1=\mathcal{C}_2=\{W_{n1}\}_{n=1}^{N}$, and $\mathcal{C}_3=\mathcal{C}_4=\{W_{n2}\}_{n=1}^{N}$. In the delivery phase, for any demand vector $\dv$, the ENs in the first cluster can send packets $\{W_{d_k1}\}_{k=1}^{4}$ cooperatively via ZF to all the $u(m)=K_R=4$ users at a time in one block, since the cluster collectively has four antennas. In a similar manner, packets $\{W_{d_k2}\}_{k=1}^{4}$ can be delivered by EN 3 and 4 to all the users in one block. The resulting NDT in \eqref{def:DE} is $\delta_E=B/F=2/2=1$.

\emph{Example 2}. Consider now a cache capacity $\mu=3/4$. In this case, the multiplicity \eqref{mmu} equals $m(\mu)=3$, which is not a divisor of $K_T$. This requires a more complex placement strategy that accounts for the need to define clusters of $m(\mu)=3$ cooperative ENs. To this end, in the proposed scheme, for caching, each file $W_{n}$ is split into $F_C=4$ disjoint parts of equal size, i.e., $W_n=\{W_{ni}\}_{i=1}^{4}$. As seen in Fig.~\ref{fig:schemee}(b), the caching policy places each part $W_{ni}$ at a contiguous cluster of $m(\mu)=3$ ENs, where contiguity is defined in a circular manner with respect to the EN index in set $\mathcal{K}_T$. More concretely, the ENs are clustered into four subsets, defined as $\mathcal{K}_{T1}=\{1,2,3\}$, $\mathcal{K}_{T2}=\{4,1,2\}$, $\mathcal{K}_{T3}=\{3,4,1\}$, and $\mathcal{K}_{T4}=\{2,3,4\}$. For any popular file $W_n$, part $W_{ni}$ is placed at all ENs in $\mathcal{K}_{Ti}$ for $i=1,2,3,4$.

Consider the worst-case request of $K_R$ distinct files. The clusters $\{\mathcal{K}_{Ti}\}_{i=1}^{4}$ of ENs are activated in turn to transmit all the requested parts $\{W_{d_ki}\}_{i=1}^{4}$ to all $k\in \mathcal{K}_R$ user. Since each EN has $n_T=2$ antennas, with EN cooperation among the three ENs in each cluster, up to $mn_T=6$ users can be served simultaneously in each block. If $K_R$ is smaller than $m n_T=6$ or a multiple thereof, it is hence possible to serve groups of $u(m)=\min\{m n_T,K_R\}$ distinct users in each block. 

Suppose now instead that we have $K_R=8$, which does not satisfy this condition. In a manner similar to the definition of the clusters of ENs, we define $B_D=4$ groups of six users $\mathcal{K}_{R1}=\{1,2,3,4,5,6\}$, $\mathcal{K}_{R2}=\{7,8,1,2,3,4\}$, $\mathcal{K}_{R3}=\{5,6,7,8,1,2\}$, and $\mathcal{K}_{R4}=\{3,4,5,6,7,8\}$. In order to serve $u(m)=6$ users simultaneously in each block, each part $W_{d_ki}$ of a requested file $W_{d_k}$ is further split into $F_D=3$ equal packets as $W_{d_ki}=\{W_{d_kij}\}_{j=1}^{3}$. For any EN cluster $\mathcal{K}_{Ti}, i\in[F_C]$, the ENs can cooperatively send a subset of $u(m)=6$ packets from $\{W_{d_ki}\}_{k=1}^{K_R}$ to all users in group $\mathcal{K}_{Rj}$ when $j=1,\cdots,B_D$, requiring $B_D$ blocks. As a result, we have $B=F_CB_D=16$ blocks and $F=F_CF_D=12$ packets, yielding the NDT $\delta_E=B/F=4/3$.

 %In each block $b \in[B_D]$, an EN cluster $\mathcal{K}_{Ti}, i\in[F_C]$ sends a subset of $u(m)=6$ packets from $\{W_{d_ki}\}_{k=1}^{K_R}$ to all users in group $\mathcal{K}_{Rb}$. 

\begin{figure}[t!] 
  \centering
\includegraphics[width=0.65\columnwidth]{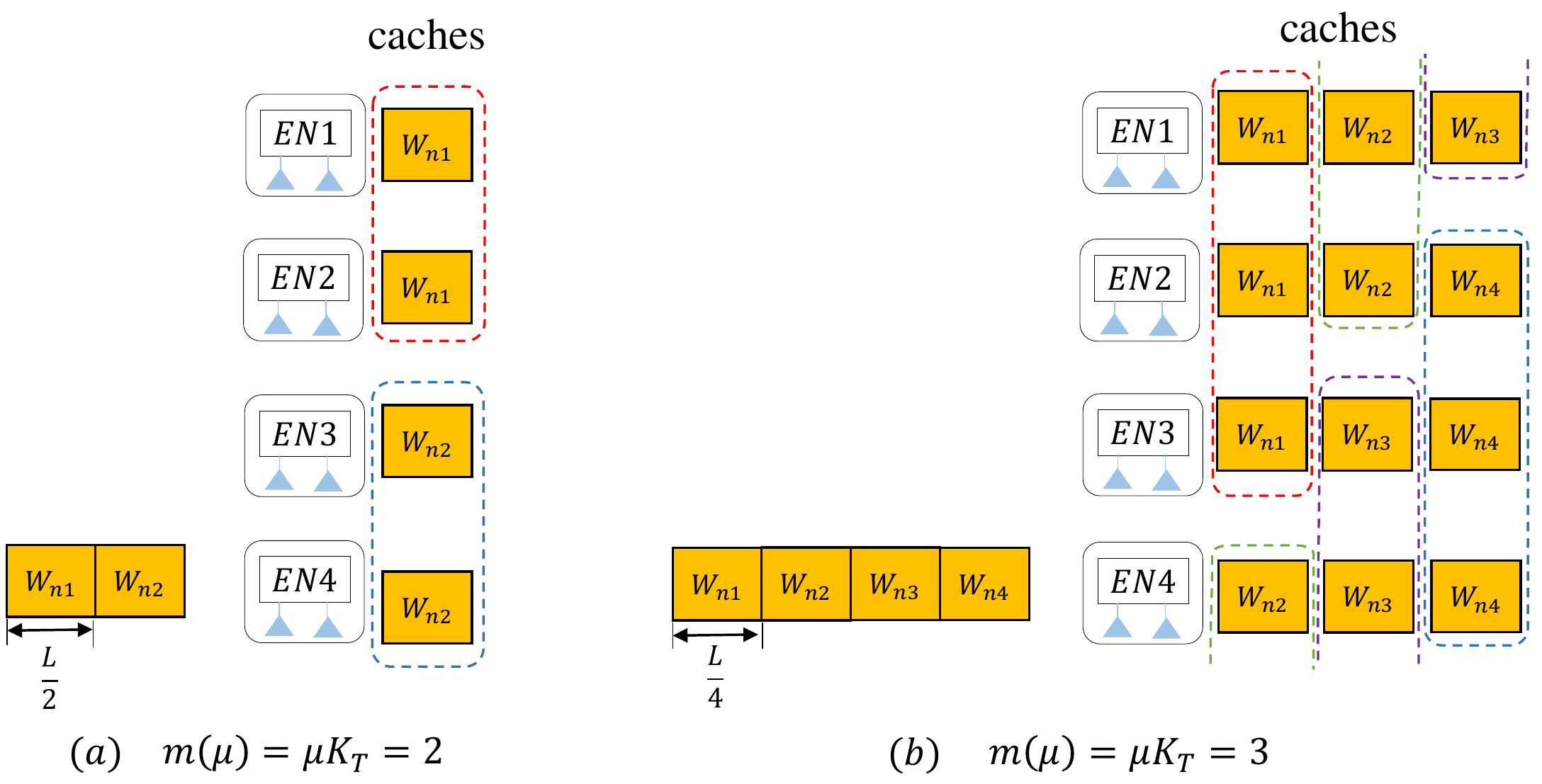}
\caption{Examples of caching scheme under edge-only transmission for the achievable NDT in Proposition~\ref{pro:NDTedge}. The dashed lines identify the clusters of cooperative ENs in \eqref{def:clusteredge}.}
\label{fig:schemee}
\end{figure}

\subsection{Proof of Proposition \ref{pro:NDTedge}} \label{sketchedge}
We now generalize the proposed scheme. For any cache capacity $\mu$, by \eqref{mmu}, we have the multiplicity $m=m(\mu)=\min\{\lfloor\mu K_T\rfloor, m_{max}\}$. To start, we define the number of parts used during the caching phase as 
\begin{align} \label{def:subedgec} 
F_C=\frac{\text{l.c.m.} (m, K_T)}{m},
\end{align}
where $\text{l.c.m.}(a,b)$ is the least common multiple of integers $a$ and $b$. As we will see, this choice guarantees that clusters of $m$ ENs can store the same part of each file, enabling cooperative transmission. We also define the number of packets created out of each cached part as
\begin{align} \label{def:subedged} 
F_D= \frac{\text{l.c.m.} (u(m), K_R)}{K_R}.
\end{align}
This ensures that in each block, subsets of $u(m)$ users can be served. Overall, the number of packets is 
\begin{align} \label{def:subedge} 
F=F_C F_D.
\end{align}
By \eqref{def:subedge}, we have the inequality $K_T/m\leq F\leq K_T K_R$, where the lower bound is attained when $K_T$ and $K_R$ are divisible by $m$ and $u(m)$, respectively. The upper bound demonstrates that the proposed scheme requires a packetization into a number of packets no larger than the product $K_T K_R$. This stands in contrast to the method of \cite{NMA:17}, which, when specialized to the case of no caching at the receivers, requires a number of packets that is exponential in the number of transmitters.

In the caching phase, each file $W_n$ is equally split into $F_C$ parts $\{W_{ni}\}_{i=1}^{F_C}$. Correspondingly, the ENs are clustered into $F_C$ clusters, defined as $\{\mathcal{K}_{Ti}\}_{i=1}^{F_C}$\footnotemark , where each cluster is defined as
\begin{align} \label{def:clusteredge}
\mathcal{K}_{Ti}=\{[(i-1)m+1]_{K_T}, [(i-1)m+2]_{K_T}, \cdots, [i m]_{K_T}\},
\end{align}
where $[a]_b=1+\bmod(a-1,b)$ for integers $a$ and $b$. Then, part $W_{ni}$ is stored in the caches of all $m$ ENs in subset $\mathcal{K}_{Ti}$.

In the delivery phase, consider a demand vector $\dv$. With cooperative ZF precoding, each cluster can serve $u(m)$ users at a time. Based on this, the $K_R$ users are grouped into 
\begin{align} \label{block}
B_D= \frac{\text{l.c.m.} (u(m), K_R)}{u(m)} 
\end{align}
groups, defined as $\{\mathcal{K}_{Rj}\}_{j=1}^{B_D}$\footnotemark[\value{footnote}]\footnotetext{The set $\{\mathcal{K}_{Ti}\}_{i=1}^{F_C}$ is a 1-($K_T,m,\text{l.c.m.} (u(m), K_R)/K_T$) design, the set $\{\mathcal{K}_{Rj}\}_{j=1}^{B_D}$ is a 1-($K_R,u(m),F_D$) design (see~\cite{DS:04}).}, with
\begin{align} \label{def:groupedge}
\mathcal{K}_{Rj} =\{[(j-1)u(m)+1]_{K_R}, [(j-1)u(m)+2]_{K_R}, \cdots,[j u(m)]_{K_R}\}.  
\end{align}
Each cluster $\mathcal{K}_{Ti} $, $i\in[F_C]$ transmits the parts $\{W_{d_ki}\}_{k=1}^{K_R}$ of the requested files by serving each of the $B_D$ groups of users in turn. To communicate to all $u(m)$ users in each group, each part $W_{d_ki}$ is further split equally into $F_D$ packets as $W_{d_ki}=\{W_{d_kij}\}_{j=1}^{F_D}$. With this split, each cluster of ENs shares $K_R F_D$ packets, which can be sent to $B_D$ groups of users sequentially within $K_R F_D/u(m)=B_D$ blocks since $u(m)$ packets are sent in each block. 
%In particular, in each block $j=1,\cdots,B_D$, $u(m)$ packets
%\begin{align} 
%\{ W_{d_kif_{jk}}: k\in \mathcal{K}_{Rj}, f_{jk}=1+\sum_{l=1}^{j-1}1(k\in \mathcal{K}_{Rl})\}
%\end{align}
%are sent, where $1()$ is an indicator function, having the value 1 if $k\in \mathcal{R}(l)$ and the value 0 if $k\notin \mathcal{R}(l)$.
%At the end, there are 
The resulting number of blocks is $B=F_C B_D$, yielding the NDT in \eqref{def:DE} $\delta_E=B/F=B_D/F_D$, as indicated in Proposition~\ref{pro:NDTedge}.  \QEDA

%Note that If $K_R$ is smaller than $u(m)$ or a multiple thereof, delivery can be carried out in a straightforward manner by serving groups of $u(m)$ users at a time, 

%To complete the sketch of the achievable scheme, we should finally consider the case when $\mu K_T$ is not an integer and we have $\mu K_T \geq m(r)$. To make full use of the edge caches, the library files are split into two disjoint fractions, which are cached with two multiplicities $\lfloor\mu K_T\rfloor$ and $\lceil\mu K_T \rceil$, respectively, and time sharing is used for the delivery of the two fractions. Further details can be found in Appendix~\ref{sec:proof}. 

\section{Achievable NDT for Fronthaul-only caching} \label{sec:fronthaul}
As a second preliminary result of interest, in this section, we present an achievable NDT (Proposition~\ref{pro:NDTft}) for the case of no caching, i.e., with $\mu=0$, as well as and the corresponding cloud-aided delivery scheme. We focus on serial fronthaul-edge transmission, while pipelined delivery will be considered in Section~\ref{sec:pipeline}.

\subsection{Achievable NDT} \label{subsec:NDTft}

In the absence of caching, any desired multiplicity level $m\leq m_{max}$ can be realized for all the contents in the requested vector $\dv$ thanks to fronthaul transmission. To this end, the fronthaul links are used to convey each packet of a requested file to a subset of $m$ ENs. Choosing the subsets of ENs as described in the previous section (see \eqref{def:clusteredge}) allows the edge NDT \eqref{ach:edge} to be achieved thanks to cooperative EN transmission. Increasing $m$ requires a larger fronthaul NDT since it requires to transfer more information on the fronthaul links, but it generally yields a lower edge NDT \eqref{ach:edge}. The next proposition presents an achievable NDT obtained by optimizing over the values of $m$.

To elaborate, define the desired multiplicity for a given fronthaul rate $r$ as 
\begin{align} \label{mmin}
m(r) = \left\{
\begin{array}{ll}
     \Big [\sqrt{\frac{K_Tr}{n_T}}~\Big], & \text{for}~r<r_{th} \\
       m_{max}, & \text{for}~ r \geq r_{th},
\end{array} 
\right.
\end{align}
where $[x]$ represents the nearest positive integer function, and
\begin{align} \label{rth}
r_{th}=\frac{n_T}{K_T}m_{max}^2. 
\end{align}
The multiplicity \eqref{mmin} is illustrated in Fig~\ref{fig:minr}. As seen, the selected multiplicity $m(r)$ is piece-wise constant and non-decreasing with respect to the fronthaul rate $r$. It is also respectively a non-decreasing an non-increasing function of $K_T$ and $n_T$. 
%The multiplicity $m(r)$ is no greater than $m_{max}$ in \eqref{mmax} in the sense that we have $u(m) = \min \{m(r)n_T, K_R\}$ and the maximum possible number of users that can be served simultaneously is again $\min\{K_T n_T, K_R\}$. $m(r)$ is a piece-wise constant non-decreasing function as $r$. The resulting achievable NDT is presented in the following proposition. 

\begin{figure}[t!] 
  \centering
\includegraphics[width=0.4\columnwidth]{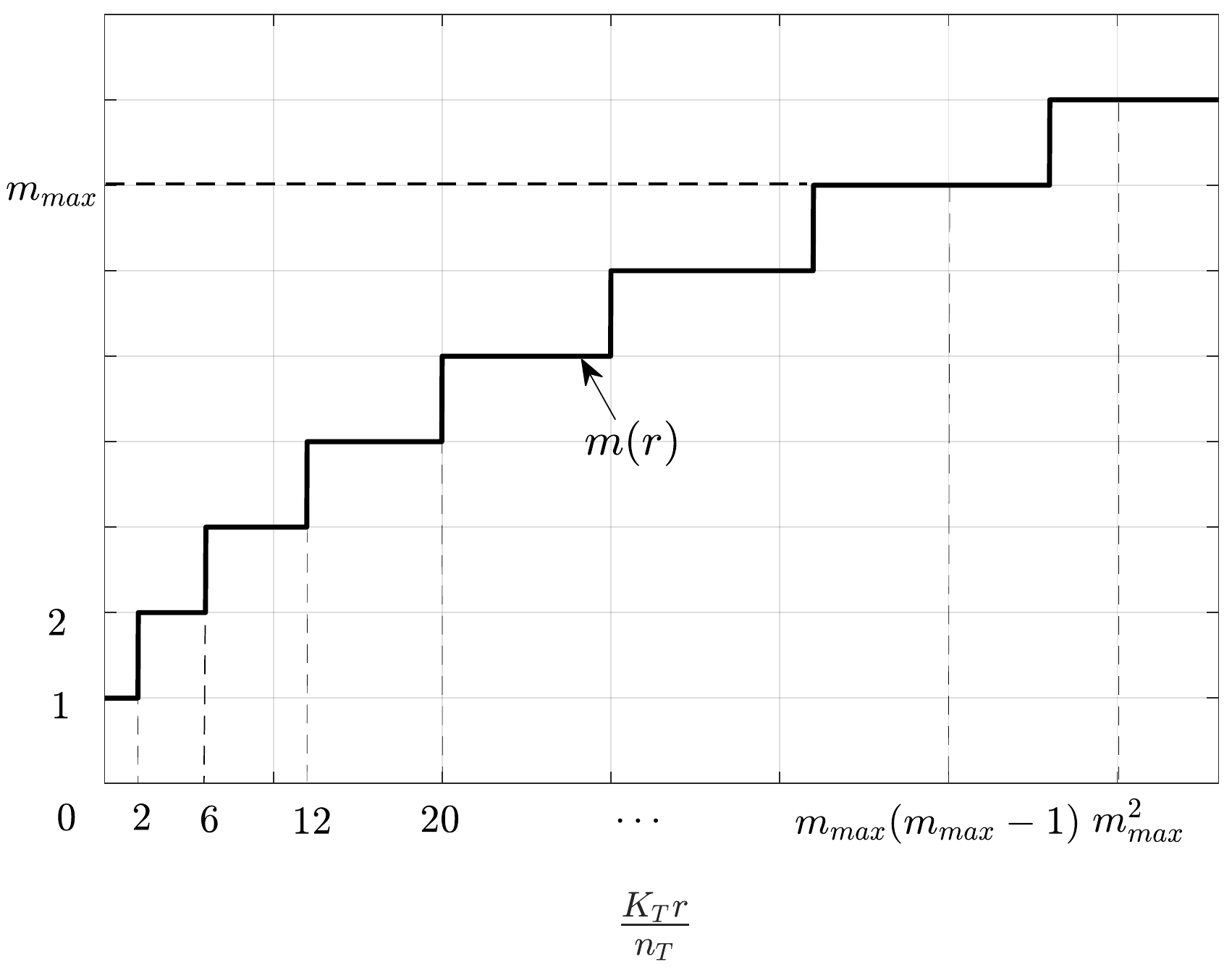}
\caption{Multiplicity $m(r)$ in \eqref{mmin} of the requested files obtained as a result of fronthaul transmission as a function of $K_Tr/n_T$.}
\label{fig:minr}
\end{figure}

\begin{proposition} \label{pro:NDTft}
For an F-RAN system with any fronthaul rate $r> 0$ and cache capacity $\mu =0$, we have the upper bound on the minimum NDT 
\begin{align} \label{ach:NDTft}
\delta^*(\mu=0,r)\leq\delta_F(m(r))+\delta_E(m(r)), 
\end{align}
with the fronthaul NDT as a function of the multiplicity $m$ given as  
\begin{align} \label{eq:deltaF}
\delta_F(m)=\frac{K_Rm}{K_T r},
\end{align}
the multiplicity $m(r)$ in \eqref{mmin}, and the edge NDT $\delta_E(m)$ defined in \eqref{ach:edge}.
\end{proposition}
\begin{proof}
The proof is presented in Section~\ref{sketchft}.
\end{proof}

\subsection{Example}

We now discuss an example that illustrates the proposed achievable cloud-aided delivery scheme. As in Example 1, we consider an F-RAN model with $K_T=4$ ENs, $n_T=2$ per-EN antennas, and $K_R=4$ users, but with no caching, i.e., with $\mu=0$ and $r>0$. 

\emph{Example 3}. Consider $r=2$, so we have the multiplicity $m(r)=2$ from \eqref{mmin}. Note that the multiplicity is the same as in Example 1 (see Fig.~\ref{fig:schemee}(a)). We hence use the same packetization and the same division of ENs into clusters of $m(r)=2$ ENs discussed in Example 1, with the caveat that here only the packets of the requested files in $\dv$ are made available to the ENs via fronthaul transmission. To elaborate, for any demand vector $\dv$, each requested file $W_{d_k}$ is divided into $F=2$ equal and disjoint packets $\{W_{d_k1}, W_{d_k2}\}$ and the ENs are clustered into two groups, namely $\mathcal{K}_{T1}=\{1,2\}$ and $\mathcal{K}_{T2}=\{3,4\}$. With fronthaul transmission, packets $\{W_{d_k1}\}_{k=1}^{K_R}$ are sent to the ENs in cluster 1, and packets $\{W_{d_k2}\}_{k=1}^{K_R}$ to the ENs in cluster 2, as illustrated in Fig.~\ref{fig:schemefonly}. Hence, we have $\mathcal{F}_1=\mathcal{F}_2=\{W_{d_k1}\}_{k=1}^{K_R}$ and $\mathcal{F}_3=\mathcal{F}_4=\{W_{d_k2}\}_{k=1}^{K_R}$, and the resulting fronthaul NDT in \eqref{def:DF} is $\delta_F=|\mathcal{F}_i|/(Fr)=4/(2\times2)=1$. For edge transmission, the cooperative delivery strategy in Example 1 can be applied by sending packets $\{W_{d_k1}\}_{k=1}^{4}$ and $\{W_{d_k2}\}_{k=1}^{4}$ sequentially in two blocks by the two clusters of ENs. Hence, the resulting edge NDT is $\delta_E=B/F=2/2=1$, yielding the overall NDT $\delta_{ach}(\mu=0,r)=\delta_F+\delta_E=2$, as in \eqref{ach:NDTft}.

\begin{figure}[t!] 
  \centering
\includegraphics[width=0.3\columnwidth]{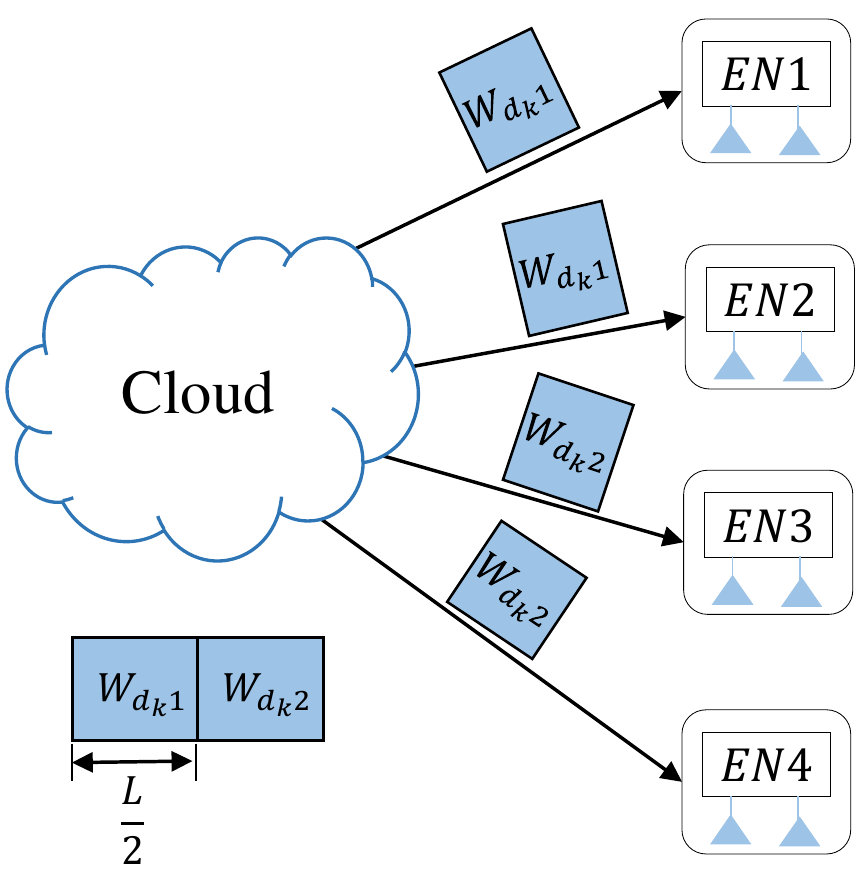}
\caption{Fronthaul transmission for the achievable NDT in Proposition~\ref{pro:NDTft} for $m(r)=2$.}
\label{fig:schemefonly}
\end{figure}

\subsection{Proof of Proposition~\ref{pro:NDTft}} \label{sketchft}

Fix a desired multiplicity $m$ for the requested contents. Each requested file $W_{d_k}$ is divided into $F_C$ parts $W_{d_k}=\{W_{d_ki}\}_{i=1}^{F_C}$ with $F_C$ in \eqref{def:subedgec}. Part $W_{d_ki}$ is sent to all the ENs in group $\mathcal{K}_{Ti}$ defined in \eqref{def:clusteredge} by using fronthaul transmission. Each part $W_{d_ki}$ is split into $F_D$ equally sized packets packets $\{W_{d_kij}\}_{j=1}^{F_D}$ with $F_D$ in \eqref{def:subedged}, and each EN $i$ receives $|\mathcal{F}_i|=K_RFm/K_T$ packets with $F=F_CF_D$ in \eqref{def:subedge}, so that the fronthaul NDT is $\delta_F(m)=|\mathcal{F}_i|/(Fr)=K_Rm/(K_Tr)$ as in \eqref{eq:deltaF}. Transmission on the edge channels takes place as described in Section~\ref{sketchedge}, entailing the NDT $\delta_E(m)$ in \eqref{ach:edge}. The overall NDT for a given multiplicity $m$ is hence given as $\delta(m)=\delta_F(m)+\delta_E(m)$, which can be minimized over $m$. To this end, we define the function 
\begin{align} \label{fun:min}
\delta(x)=\frac{K_R x}{K_T r}+\frac{K_R}{x n_T},
\end{align}
where $x\in[0,m_{max}]$ is a variable obtained by relaxing the integer constraints over $m$. The function $\delta(x)$ is convex within the range $[0,m_{max}]$, and the only stationary point is $x_0 =\sqrt{K_T r/n_T}$. Therefore, function $\delta(x)$ reaches its minimum at $x=x_0$. Based on this, the optimal multiplicity $m$ is either $\lfloor x_0 \rfloor$ or $\lceil x_0 \rceil$ depends on whether $\delta(\lfloor x_0 \rfloor)<\delta(\lceil x_0 \rceil)$ or $\delta(\lfloor x_0 \rfloor)>\delta(\lceil x_0 \rceil)$, respectively. Hence, to simplify the analysis, the desired multiplicity $m$ is chosen as the nearest positive integer of $x_0$, although this may not be optimal. This completes the proof. \QEDA

\section{Normalized delivery time Analysis on F-RAN} \label{sec:edgeft}

In Section~\ref{sec:edge} and~\ref{sec:fronthaul}, we have studied the special cases with edge caching only and no caching, respectively. In this section, we proceed to consider a general F-RAN model with any fronthaul rate $r\geq 0$ and cache capacity $\mu \geq 0$. We present an upper bound (Proposition~\ref{pro:NDT}) and a lower bound (Proposition~\ref{pro:lower}) on the minimum NDT under serial delivery. These bounds provide a characterization of the minimum NDT that is conclusive for a wide range of values of the system parameters (Proposition~\ref{pro3}) and is generally within a multiplicative factor of $3/2$ from optimality (Proposition~\ref{pro:gap}). The main results offer insight into the optimal use of cloud and edge resources as a function of the fronthaul capacity, cache resources and number of ENs' transmit antennas. 

\subsection{Upper Bound on the Minimum NDT} \label{sec:sketch}

%In this subsection, we describe a caching and delivery scheme and the corresponding achievable NDT. 
In the presence of both cloud and edge resources, the multiplicity $m$ of the requested files depends both on the pre-stored caching contents and on the information received at the ENs via fronthaul transmission.  As a result, in an F-RAN,  the optimal multiplicity is generally larger than the multiplicities $m(r)$ and $m(\mu)$ in \eqref{mmu} and \eqref{mmin}, respectively. 
%Recall that when $\mu=0$, the multiplicity $m(r)$ is selected, while when $r=0$, the multiplicity $m(\mu)=\lfloor\mu K_T\rfloor$ can be ensured by the ENs' cached contents. Instead, when both cloud and edge resources are offered, i.e., when $\mu \geq 0, r \geq 0$, a multiplicity, which is bigger than both $m(r)$ and $m(\mu)$, can be supported by sending information on the requested files from the cloud. However, this comes with a cost of fronthaul transmission latency. To balance this cost and the benefits from cooperation of ENs, 
The multiplicity selected in the proposed scheme for any values of $\mu$ and $r$ is given as 
\begin{align} \label{m}
m(\mu,r) = \left\{
\begin{array}{ll}
       m(r), & \text{if}~ \mu K_T  < m(r) \\
			 \lfloor \mu K_T \rfloor, & \text{if}~ m(r) \leq  \mu K_T\leq   m_{max} \\
       m_{max}, & \text{if}~ \mu K_T > m_{max}. 
\end{array} 
\right.
\end{align}
The formula \eqref{m} has a graphical interpretation that will be discussed after the next proposition (see Fig.~\ref{fig:m}). Note that we have the equalities $m(\mu,0)=m(\mu)$ and $m(0,r)=m(r)$.

\begin{figure}[t!] 
  \centering
\includegraphics[width=0.45\columnwidth]{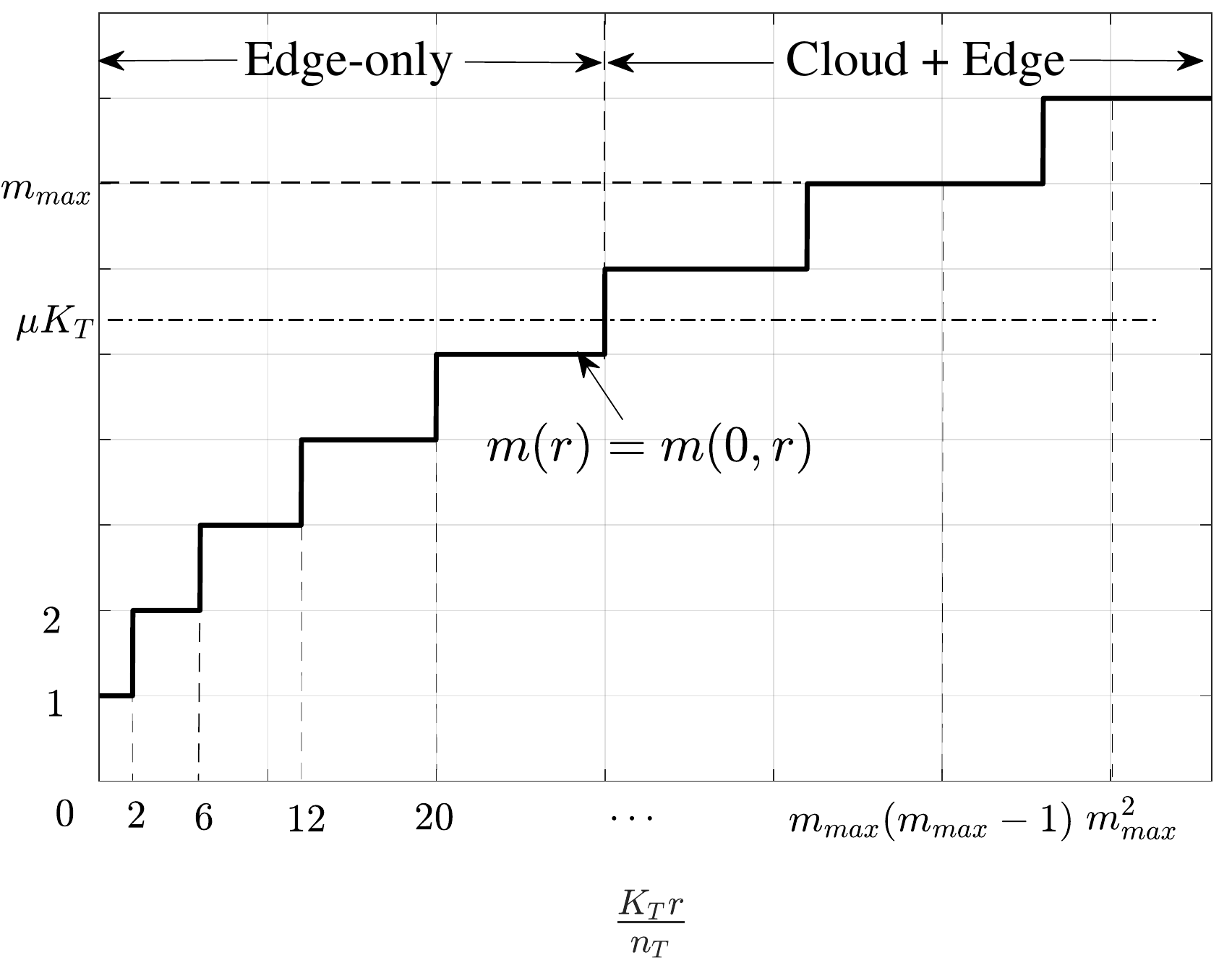}
\caption{Illustration of the multiplicity $m(\mu,r)$ \eqref{m} of the requested files, as a result of caching and fronthaul transmission, as selected by the proposed scheme: for $\mu K_T \leq m_{max}$, this is obtained by taking the maximum between $\lfloor \mu K_T \rfloor$ and $m(r)$ for the given value of $r$; while $ \mu K_T> m_{max}$, we have $m(\mu,r)=m_{max}$.}
\label{fig:m}
\end{figure}

\begin{proposition} \label{pro:NDT}
For an F-RAN system with any fronthaul rate $r\geq 0$ and cache capacity $\mu\geq 0$, we have the upper bound on the minimum NDT 
%\begin{align} \label{ach:smallr}
%\delta_{ach}(\mu,r) = \left\{
%\begin{array}{ll}
      %\delta_F(\mu,r)+\delta_E(\mu,r), & \text{for}~\mu K_T \leq m(r) \\ 
		%%\frac{K_R(m(\mu,r)+1-\mu K_T)}{u(\mu,r)}+\frac{K_R(\mu K_T-m(\mu,r))}{u'(\mu,r)}, & \text{for}~   \mu K_T \geq m(\mu,r). \\ [1.5ex]
		 	%%\alpha \delta_E(\mu,r)+(1-\alpha) \delta'_E(\mu,r), & \text{for}~   \mu K_T \geq m(r), 
			%\delta_E(\mu,r), & \text{for}~   \mu K_T \geq m(r), 
%\end{array} 
%\right.
%\end{align}
\begin{align} \label{ach:smallr}
\delta^*(\mu,r)\leq \delta_{ach}(\mu,r)= \text{l.c.e.}(\delta_F(\mu,r)+\delta_E(\mu,r)),
\end{align}
with the fronthaul NDT 
\begin{align} \label{eq:deltaEF}
\delta_F(\mu,r)=\delta_F(m(\mu,r)-\lfloor \mu K_T\rfloor),
\end{align}
with $\delta_F(m)$ in \eqref{eq:deltaF}, and the edge NDT
\begin{align} \label{eq:deltaEE1}
 \delta_E(\mu,r)=\delta_E(m(\mu,r)),
\end{align}
with $\delta_E(m)$ in \eqref{ach:edge}.
%\item for any $r\geq  r_{th}$,
%\begin{align}
%\delta_{ach}(\mu,r)=\frac{K_R (m_{max}-\mu K_T)^+}{ K_Tr} +\frac{K_R}{\min\{ K_Tn_T, K_R\}}. \label{NDT:bigr}
%\end{align}
\end{proposition}
\begin{proof}
The proof is presented in Appendix~\ref{sec:proof}. 
\end{proof}

\begin{figure}[t!] 
  \centering
\includegraphics[width=0.45\columnwidth]{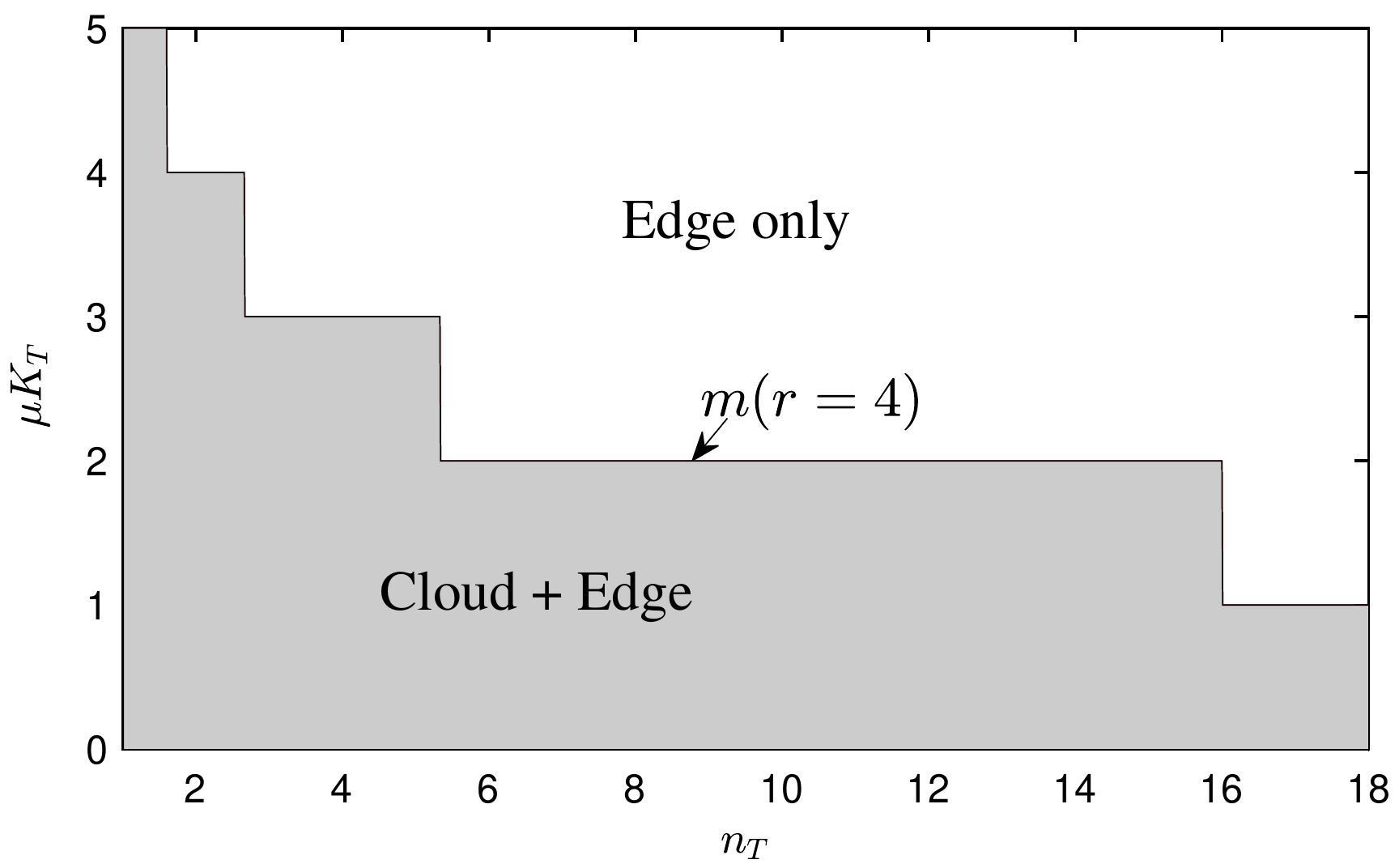}
\caption{The proposed caching-delivery scheme, which is approximately optimal by Proposition 3, leverages the cloud resources only for values of $\mu$ and $n_T$ in the shaded area, while edge transmission is exclusively used otherwise ($K_T=8,K_R=40$ and $r=4$).}
\label{fig:munt}
\end{figure}

According to \eqref{m}, as illustrated in Fig.~\ref{fig:m}, the multiplicity $m(\mu,r)$ of each requested file is obtained by comparing $ \mu K_T $, i.e., the multiplicity allowed by caching only, with the upper and lower bound $m_{max}$ and the optimal multiplicity $m(r)$ selected when $\mu=0$. We distinguish two cases.
\begin{itemize}
	\item \emph{Edge-only transmission:} For a sufficiently large cache capacity, i.e., when $\mu K_T \geq m(r)$, by \eqref{mmu} and \eqref{m}, we have $m(\mu,r)=m(\mu)$, and hence the edge caches can support the selected multiplicity without the need for fronthaul transmission. In this case, we have $\delta_F(\mu,r)=0$, and the NDT $\delta_{ach}(\mu,r)$ in \eqref{ach:smallr} includes only the edge contribution $\delta_E(\mu,r)$;
 \item \emph{Cloud and edge-aided transmission:} When $\mu K_T < m(r)$, we have the multiplicity $m(\mu,r)= m(r) > \mu K_T$, and hence fronthaul transmission is needed in order to support the multiplicity $m(\mu,r)$. In this case, the NDT $\delta_{ach}(\mu,r)$ in \eqref{ach:smallr} includes both the contributions of fronthaul and edge NDTs. 
\end{itemize}

\begin{remark} \label{remark}
From the discussion above, whenever the cache capacity is small enough to satisfy the inequality $\mu K_T < m(r)$, the proposed policy uses both cloud and edge resources. Accordingly, even when the edge alone would be sufficient to deliver all requested contents, that is, even when we have $\mu K_T\geq 1$, the policy uses cloud-to-edge communications if $r$ is sufficiently large. This is because, in this regime, the cloud can send the uncached information to ENs in order to increase the multiplicity and hence to foster EN cooperation, at the cost of a fronthaul delay that does not offset the cooperation gains. However, when $\mu K_T \geq m(r)$, the scheme only uses edge resources. In fact, under this condition, the gains due to enhanced EN cooperation do not overcome the latency associated with fronthaul transmission. Fig.~\ref{fig:munt} illustrates the discussed conditions by depicting as shaded region of values of the pair $(\mu,n_T)$ for which inequality $\mu K_T < m(r)$ is satisfied and hence both cloud and edge transmission is used by the proposed scheme. An interesting observation is that, as $n_T$ increases, edge processing becomes more effective, make cloud processing unnecessary for small values of the cache capacity $\mu$. Note also that an increased $r$ enlarges the region of values $(\mu,n_T)$ for which fronthaul transmission is used (not shown).   \QEDA
\end{remark}

\begin{figure}[t!] 
  \centering
\includegraphics[width=0.38\columnwidth]{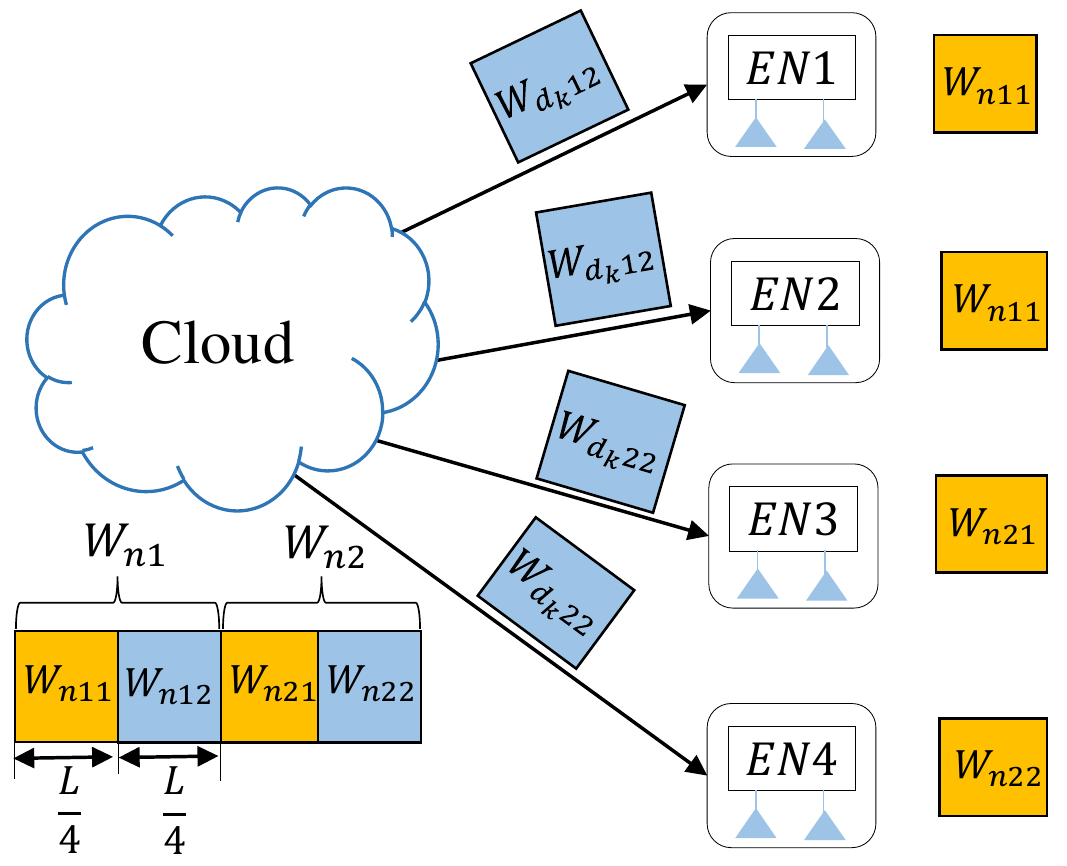}
\caption{Caching and delivery scheme under cloud and cache-aided transmission for the achievable NDT in Proposition~\ref{pro:NDT} for $m(\mu,r)=m(r)=3$.}
\label{fig:schemef}
\end{figure}

\subsection{Example} 
Here we continue Example 1 and Example 3 by considering again an F-RAN with the example with $K_T=4$ ENs, $n_T=2$ per-EN antennas and $K_R=4$ users, but in the general case with $\mu\geq0$ and $r\geq0$.

\emph{Example 4}. Consider $\mu=0.25$ and $r=2$, so we have the multiplicity $m(\mu,r)=m(r)=2$ by \eqref{m}, which is the same as in both Example 1 and Example 3. To realize this multiplicity, we carry out first the same partition $\{W_{n1}, W_{n2}\}$ of each library file $W_n$ into two $F_C=2$ parts as in Example 1 and Example 3. Moreover, here, each part is further split into $F_S=2$ disjoint packets $\{W_{ni1}, W_{ni2}\}$ of equal size. In the placement phase, only packets $\{W_{ni1}\}_{n=1}^{N}$ for all the contents $\{W_n\}_{n=1}^{N}$ are cached at the ENs in cluster $\mathcal{K}_{Ti}=\{1,2\}$ for $i=1,2$, as seen in Fig.~\ref{fig:schemef}. In the delivery phase, for any demand vector $\dv$, the uncached packets $\{W_{d_ki2}\}_{k=1}^{K_R}$ of requested files are sent to the ENs in cluster $\mathcal{K}_{Ti}$. Therefore, each EN receives four packets on the fronthaul, yielding the fronthaul NDT $\delta_F=|\mathcal{F}_i|/(Fr)=4/(4\times2)=1/2$. Using clustered EN cooperation, packets $\{W_{d_k 1i}\}_{k=1}^{K_R}$ for $i=1,2$ can be sent by the ENs in cluster $\mathcal{K}_{T1}$ in two blocks, while packets $\{W_{d_k2i}\}_{k=1}^{K_R}$ can be similarly delivered by the ENs in cluster $\mathcal{K}_{T2}$. As a result, the edge NDT is $\delta_E=B/F=4/4=1$. Hence, the overall NDT is $\delta_{ach}(\mu,r)=\delta_F+\delta_E=3/2$, as in \eqref{ach:smallr}.

\subsection{Lower Bound on the Minimum NDT}
A lower bound on the minimum $\delta^*(\mu, r)$ is presented in the following 
proposition, where we define the function $m^*(r)$ as 
\begin{align} \label{mmins}
m^*(r) = \left\{
\begin{array}{ll}
       \max\Big\{\sqrt{\frac{K_T r}{n_T}},1\Big\}, & \text{for}~ r<r_{th} \\ 
       m_{max}, & \text{for}~ r \geq r_{th},
\end{array} 
\right.
\end{align}
with $r_{th}$ as in \eqref{rth}.

\begin{proposition} \label{pro:lower}
In an F-RAN with $n_T$ antennas at each transmitter, the minimum NDT $\delta^*(\mu, r)$ is lower bounded as
\begin{subequations} \label{eq:pro2}
  \begin{align}[left ={\delta^*(\mu, r) \geq \delta_{lb}(\mu, r)= \empheqlbrace}]
\max \mbox{\small\( \Big\{\frac{K_R(m^*(r)-\mu K_T)}{K_T r}+\frac{K_R}{m^*(r) n_T}, 1\Big\} \)},  	\normalsize  ~& \text{for}~\mu K_T < m^*(r) \label{eq:pro21}\\
 \max \mbox{\small\(  \Big\{\frac{K_R}{\mu K_T n_T},1\Big\}\)},~~~~~~~~~~~~~~~~~~~~~~~  	\normalsize  &  \text{for}~ \mu K_T \geq m^*(r).  \label{eq:pro22}
   \end{align}
  \end{subequations}
\end{proposition}
\begin{proof}
The proof is presented in Appendix~\ref{converse}. 
\end{proof}

\subsection{Minimum NDT}
The following proposition characterizes the minimum NDT $\delta^*(\mu, r) $ for the regime of low cache and fronthaul capacities, namely, when $\mu K_T \in [0,1]$ and $r\in [0,n_T/K_R]$, as well as for any set-up with $\mu K_T$ integer, or with sufficiently large caches such that $\mu K_T \geq m_{max}$.
\begin{proposition}  \label{pro3}
For an F-RAN system with $n_T$ antennas at each EN, the minimum NDT $\delta^*(\mu, r) $ is given as
\begin{align}  \label{eq:pro3}
\delta^*(\mu, r) = \left\{
\begin{array}{ll}
       \max\big\{\frac{K_R(1-\mu K_T)}{K_T r}+\frac{K_R}{ n_T}, 1\big\},  & \text{for}~\mu K_T \in [0,1]~\text{and}~r\in [0,\frac{n_T}{K_T}] \\ 
      \max\big\{\frac{K_R}{\mu K_T n_T},1\big\}, &  \text{for}~ \mu K_T\in \{m(r)+1,\cdots,m_{max}\} \cup (m_{max},K_T]. 
\end{array} 
\right.
\end{align}
\end{proposition}
\begin{proof}
The result follows by the direct comparison of the bounds in Proposition 1 and Proposition 2.
\end{proof}

\begin{figure}[t!] 
  \centering
\includegraphics[width=0.45\columnwidth]{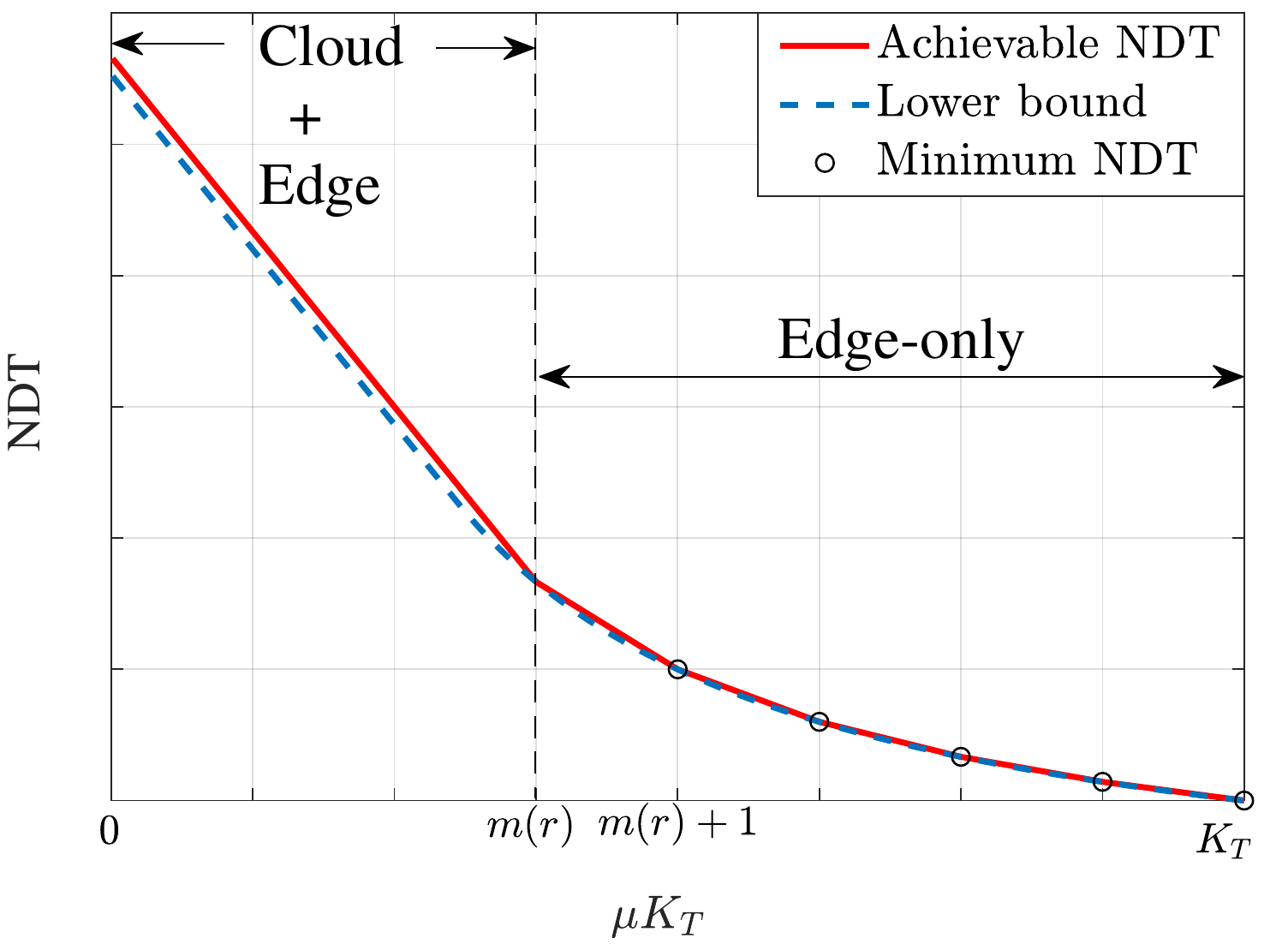}
\caption{Achievable NDT $\delta_{ach}(\mu, r)$ in Proposition~\ref{pro:NDT} (solid curve), and lower bound on the minimum NDT $\delta^*(\mu, r)$ in Proposition 2 (dashed line) versus $\mu$ for a given value of $r$. The figure highlights the two regimes of values of the cache capacity $\mu$ with which the achievable schemes use edge-only or both cloud and edge transmission. }
\label{fig:NDT}
\end{figure}

More generally, the achievable NDT in Proposition \ref{pro:NDT} is within a factor of $3/2$ from minimum NDT for any fractional caching size $\mu$ and fronthaul rate $r$.
\begin{proposition} \label{pro:gap}
For an F-RAN system with $n_T$ antennas at each EN, and any value of $\mu\geq 0$ and $r \geq 0$, we have the inequality
\begin{align} 
\frac{\delta_{ach}(\mu, r) }{\delta^*(\mu, r) } \leq \frac{3}{2} .
\end{align}
\end{proposition}
\begin{proof}
The proof is presented in Appendix \ref{gap}.
\end{proof}

A plot of the achievable NDT $\delta_{ach}(\mu, r)$ and of the lower bound $\delta_{lb}(\mu, r)$ as a function of the fractional cache capacity $\mu$ is shown in Fig.~\ref{fig:NDT} for a given value of $r$. As discussed, the achievable scheme uses both cloud and edge resources when $\mu K_T<m(r)$, while it uses only edge transmission when $\mu K_T \geq m(r)$. In the first regime, the fronthaul NDT decreases linearly with $\mu K_T$, which leads to a linear decrease in the overall NDT $\delta_{ach}(\mu, r)$. Instead, in the second regime, the achievable NDT $\delta_{ach}(\mu, r)$ is piece-wise linear and decreasing. For this range of values of $\mu$, time-sharing between two successive multiplicities is carried out for delivery, unless $\mu K_T$ is an integer. 
By comparison with the lower bound, the figure also highlights the regimes, identified in Proposition~\ref{pro3}, in which the scheme is optimal. 

\begin{figure}[t!] 
  \centering
\includegraphics[width=0.45\columnwidth]{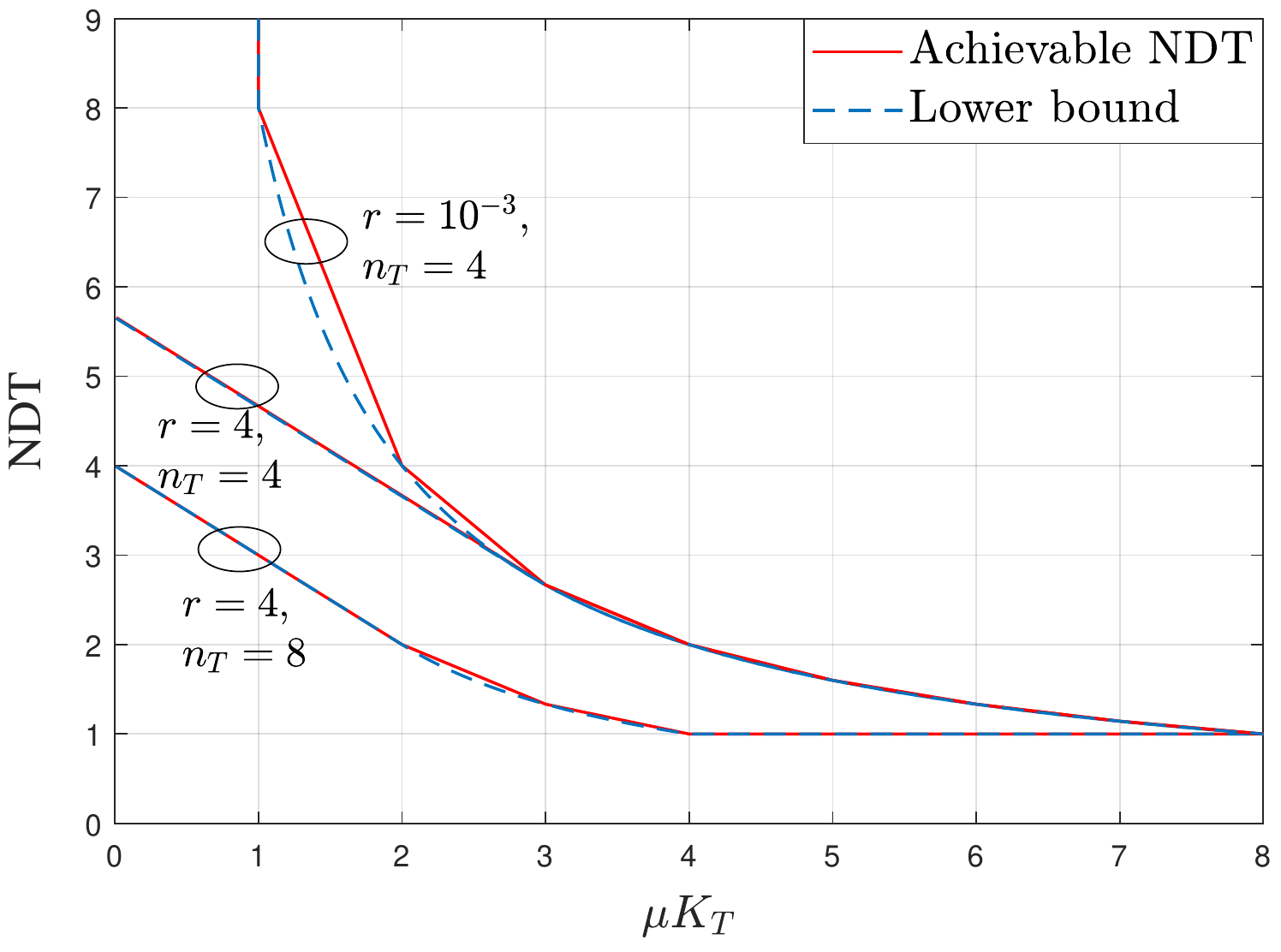}
\caption{Achievable NDT $\delta_{ach}(\mu, r)$ and lower bound $\delta_{lb}(\mu, r)$ versus $\mu$ for different values of $r$ and $n_T$, with $K_T=8$ and $K_R=32$.}
\label{fig:NDTrs}
\end{figure}

The achievable NDT in Proposition 1 and lower bound in Proposition 2 are plotted in Fig. \ref{fig:NDTrs} as a function of $\mu$ for $K_T=8$ and $K_R=32$ and for different values of $r$ and $n_T$. As stated in Proposition 3, the achievable NDT is optimal when $\mu$ and $r$ are small enough, as well as when $\mu$ equals a multiple of $1/K_T=1/8$ or is large enough. For values of $r$ close to zero, the NDT diverges as $\mu$ tends to $1/K_T=1/8$, since requests cannot be supported based solely on edge transmission. For larger values of $r$ and/or $n_T$, the NDT decreases. In particular, when $\mu K_T\geq m_{max}=4$ and $n_T=8$, as discussed, we have the ideal NDT of one, since the maximum possible number $32$ of users can be served.

\begin{figure}[t!] 
  \centering
\includegraphics[width=0.45\columnwidth]{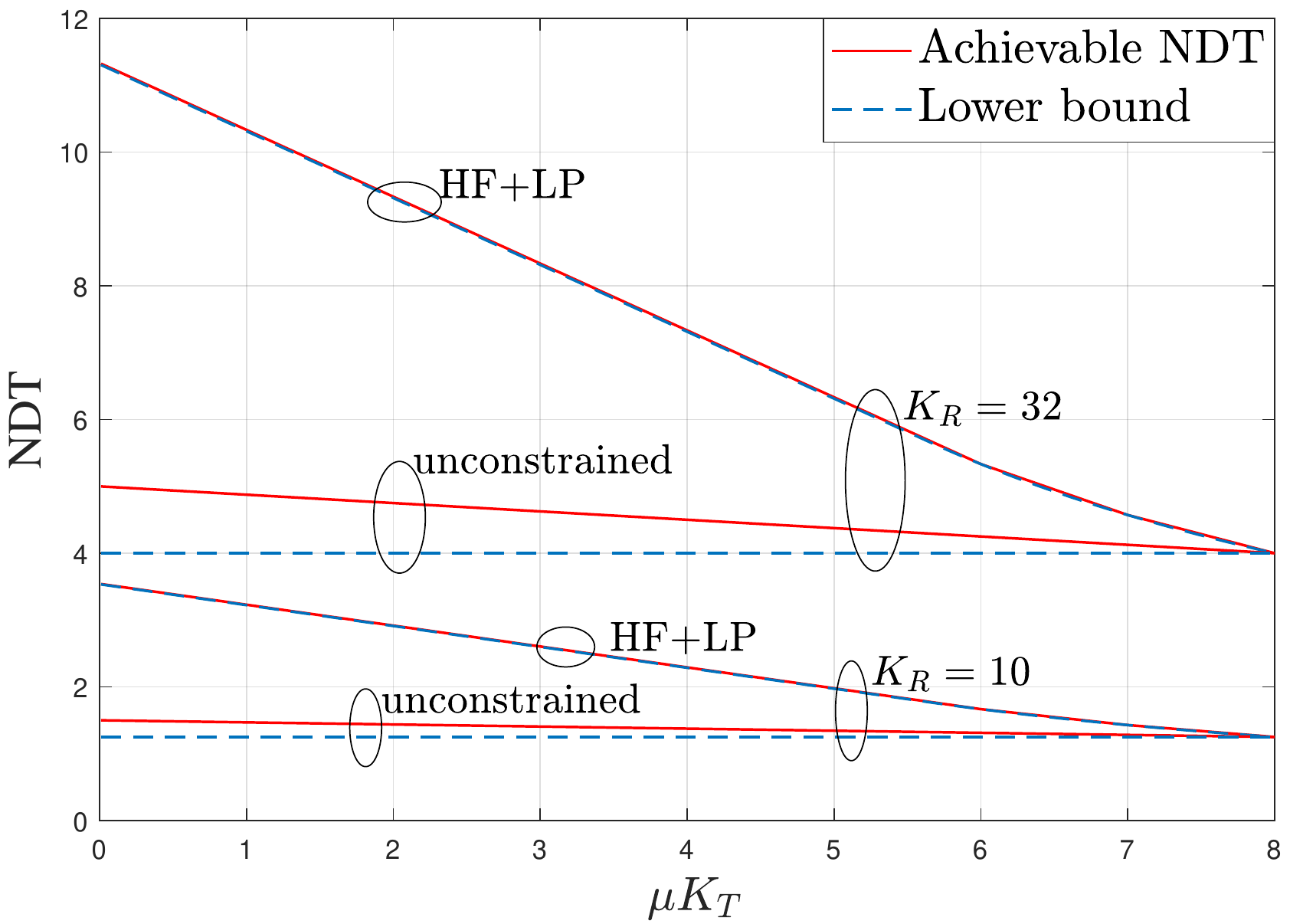}
\caption{Comparison of the bounds derived in this paper under the constraints of hard-transfer fronthauling (HF) and one-shot linear precoding (LP) with the bounds derived in \cite{STS:17} without such constraints, with $K_T = 8, n_T = 1$ and $r = 4$.}
\label{fig:NDTrsj}
\end{figure}

Finally, for reference, a comparison of the bounds derived here under the assumptions of hard-transfer fronthauling (HF), i.e., the transmission of uncoded files on the fronthaul links, and of one-shot linear precoding (LP) with those derived in [7, Corollary 1 and Proposition 4] without such constraints, is illustrated in Fig.~\ref{fig:NDTrsj}. We recall that the achievable scheme in \cite{STS:17} is based on fronthaul quantization and on delivery via interference alignment and ZF precoding. The figure is obtained for $K_T=8, n_T=1, r=4$, and for different values of $K_R$. It is observed that the loss in performance caused by the practical constraints considered in this work is significant, and that it increases with the number $K_R$ of users. This conclusion confirms the discussion in [7, Sec. IV-B].

\section{Pipelined fronthaul-edge transmission} \label{sec:pipeline}
In this section, we consider pipelined fronthaul-edge transmission, whereby fronthaul transmission from the cloud to the ENs and edge transmission from the ENs to the users can take place simultaneously. Note that this is possible due to the orthogonality of the two channels. We first describe how the system model and performance metric are modified as compared to the serial model of Section~\ref{sec:model}, and then we present upper and lower bounds on the minimum NDT. As for the serial case, the bounds will reveal that the proposed cloud and cache-aided transmission policy is optimal for a large range of system parameters, and is generally within a multiplicative factor, here of two, from the information-theoretic optimal performance. Importantly, in contrast to the approximately optimal serial strategy, the proposed scheme for pipelined transmission leverages fronthaul transmission for any non-zero value of the fronthaul rate $r$ and for any value of the fractional cache capacity $\mu$ less than one. 

\subsection{System Model and Performance Metric}
The system model is as in Section~\ref{sec:model}, with the main caveat that fronthaul and wireless transmissions can occur simultaneously. Specifically, caching is defined as in Section~\ref{sec:model} by sets $\{\mathcal{F}_{i}\}_{i=1}^{K_T}$. As shown in Fig.~\ref{fig:pipeline}, the overall delivery time for a given request vector is organized into $B$ blocks. In any block $b \in [B]$, the cloud transmits the packets in set $\mathcal{F}_i(b)=\{\mathcal{F}_{id_1}(b),~\cdots,~\mathcal{F}_{id_K}(b)\}$ to EN $i$, where $\mathcal{F}_{id_k}(b)\in \mathcal{F}_{id_k}$ is a subset of the packets requested by user $k$. In the same block, each EN $i$ sends the subset $\mathcal{D}(b)$ of requested packets to a subset $\mathcal{R}(b)$ of users by utilizing the cached contents and fronthaul information $\{\mathcal{F}_i(b')\}_{i=1}^{K_T}$ received in the previous blocks $b'=1,~\cdots,~b-1$. Note that the ENs can use the information received on the fronthaul in a causal way along the blocks. As in \eqref{def:TF}, the duration of the fronthaul transmission in each block $b$ is given as
\begin{align}
T_F(b) = \max_{i\in[K_T]}\frac{|\mathcal{F}_i(b)|L}{F}\frac{1}{C_F}, \label{def:TFP}
\end{align}
and, following \eqref{def:TE}, the edge transmission time is given as the sum over the blocks
\begin{align}
T_E(b)=\frac{L}{F}\frac{1}{(\log(P) + o(\log(P))}. \label{def:TEP}
\end{align}
Since each block needs to accommodate both fronthaul and edge transmissions, the duration of a block is the maximum of the above two times, i.e., $T_P(b)=\max\{T_F(b),T_E(b)\}$. The total delivery time is hence given as 
\begin{align}\label{pt}
T_P=\sum_{b=1}^{B} T_P(b)=\sum_{b=1}^{B} \max\{T_E(b), T_F(b)\}.
\end{align}
Finally, following \eqref{def:DF}-\eqref{def:DE}, the pipelined NDT $\delta_{P}$ is computed as the limit
\begin{align}\label{def:deltap}
\delta_P =\lim_{P\rightarrow\infty} \lim_{L\rightarrow\infty}  \frac{T_P}{L/\log(P)}.
\end{align}
The minimum NDT $\delta_P^*(\mu,r)$ is defined as in \eqref{def:deltabar}-\eqref{def:delta}. Following the same argument in [7, Lemma 4], the minimum NDT $\delta_P^*(\mu,r)$ pipelined delivery satisfies the inequalities $\delta^*(\mu,r)/2 \leq \delta_P^*(\mu,r)\leq \delta^*(\mu,r)$, and hence the improvement in NDT under pipelined transmission is upper bounded by a factor of two.

\begin{figure}[t!] 
  \centering
\includegraphics[width=0.55\columnwidth]{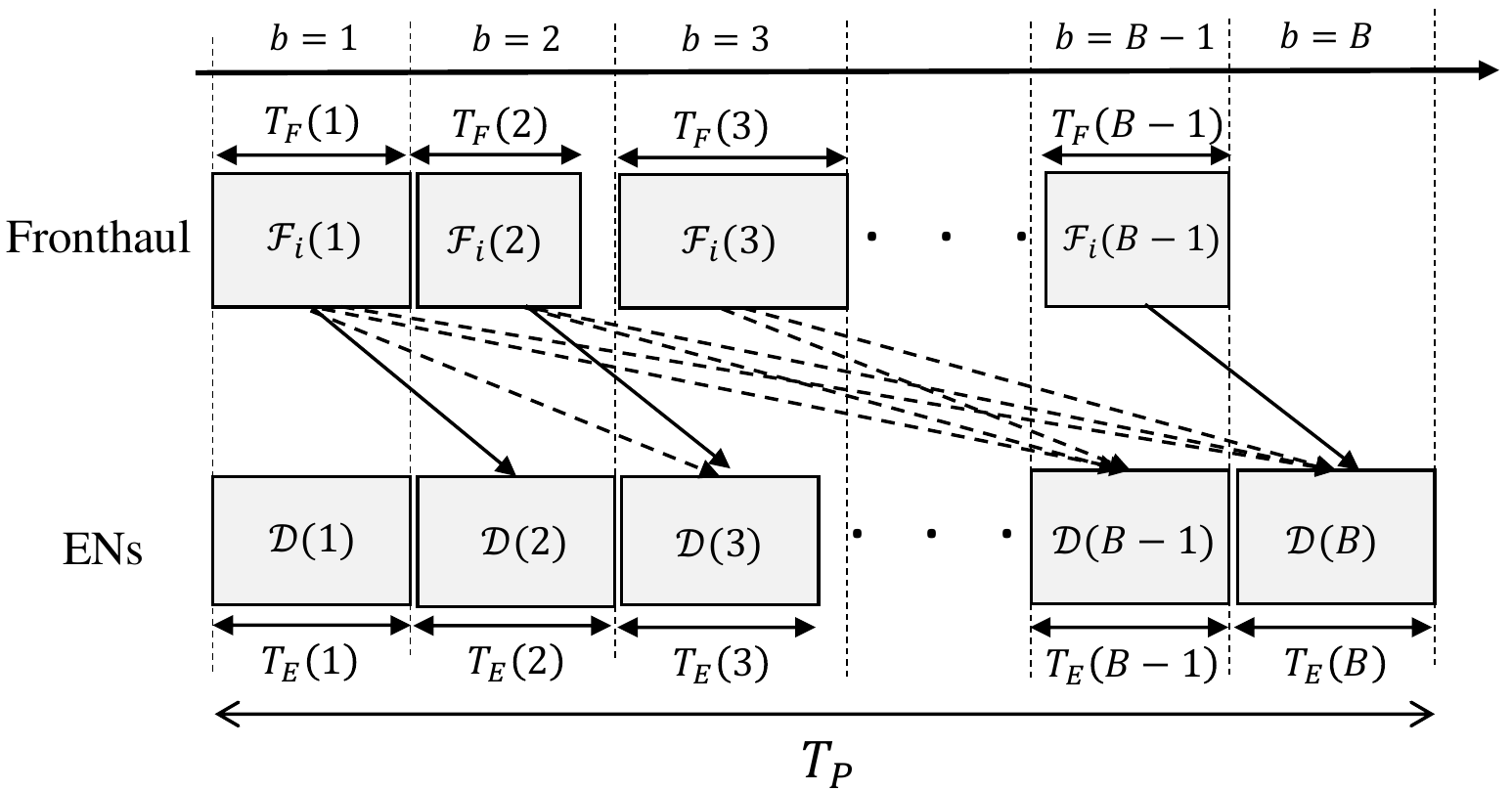}
\caption{Illustration of pipelined F-RAN operation for general strategies (both dashed and solid lines) and for block-Markov transmission (solid lines only).}
\label{fig:pipeline}
\end{figure}

%From the above discussion, for the optimal pipelined strategy corresponding to the NDTs $\delta_E^*(\mu,r)$ and $\delta_F^*(\mu,r)$, the minimum pipelined NDT satisfies $\delta_{P}^*(\mu,r)\geq \max\{\delta_E^*(\mu,r),\delta_F^*(\mu,r)\}$. We note that, any achievable pipelined strategy can be easily converted into a serial strategy by doing fronthaul transmission and edge transmission sequentially. Hence, the NDT of the serial strategy converted from the optimal pipelined strategy is given as $\delta(\mu,r)=\delta_E^*(\mu,r)+\delta_F^*(\mu,r)$, which may not be optimal under serial transmission, i.e., $\delta(\mu,r) \geq \delta^*(\mu,r)$ (see \eqref{def:delta}). As a result, we can have 
%\begin{align}
%\frac{\delta^*(\mu,r)}{\delta_{P}^*(\mu,r)} \leq \frac{\delta(\mu,r)}{\delta_{P}^*(\mu,r)}  \leq \frac{\delta_E^*(\mu,r)+\delta_F^*(\mu,r)}{\max\{\delta_E^*(\mu,r),\delta_F^*(\mu,r)\}} \leq  \frac{2\max\{\delta_E^*(\mu,r),\delta_F^*(\mu,r)\}}{\max\{\delta_E^*(\mu,r),\delta_F^*(\mu,r)\}} =2.
%\end{align}
%This indicates that no matter how, the improvement in NDT under pipelined transmission is no larger than a factor of two (see also [6, Lemma 4]). 

\subsection{Achievable Scheme and Upper Bound on the Minimum NDT}
In this section, we derive an achievable NDT by proposing a caching and delivery strategy that leverage simultaneous fronthaul and edge transmission. We start by observing that any serial strategy defined by the tuple $\{\mathcal{C}_i, \mathcal{F}_i, \{\vv_{inf}(b)\}_{n\in[N], f\in[F],b\in[B]}\}_{i=1}^{K_T}$, as described in Section~\ref{sec:model}, can be converted into a pipelined transmission strategy $\{\mathcal{C}_i, \{\mathcal{F}_i(b)\}_{b\in[B]}, \{\vv_{inf}(b)\}_{n\in[N], f\in[F],b\in[B]}\}_{i=1}^{K_T}$. This is done by means of block-Markov transmission, as illustrated in Fig.~\ref{fig:pipeline}. To this end, caching is carried out in the same way as in the serial strategy. For delivery, any packet in the fronthaul message $\mathcal{F}_i$ is split into $B-1$ equal subpackets, and each $b$th subpacket is placed in the set $\mathcal{F}_i(b)$ communicated in block $b=1,2,\cdots,B-1$. The edge delivery scheme for the serial strategy is applied in each following block $b=2,\cdots,B$ to the subpackets received in the previous block $b-1$.

Suppose that the original serial scheme has fronthaul and edge NDTs $\delta_F$ and $\delta_E$, respectively. Then, by the definitions \eqref{def:TFP} and \eqref{def:TEP}, the contribution to the NDT for each block $b$ is given as $\max\{\delta_F,\delta_E\}/(B-1)$, since each block delivers a fraction $1/(B-1)$ of each packet. Hence, the overall NDT is given as $\delta_P=(B/B-1)\max\{\delta_F, \delta_E\}$, which yields for $B \rightarrow \infty$ the NDT 
\begin{align} \label{def:deltap}
\delta_P = \max\{\delta_F,\delta_E\}.
\end{align}

Based on the described block-Markov approach, as a first solution, one could convert the approximately optimal serial policy derived in the previous section to obtain an achievable NDT for the pipelined model. However, in the proposed serial strategy, it can be proved that the edge NDT $\delta_E(\mu,r)$ in \eqref{eq:deltaEE1} is generally larger than the fronthaul NDT $\delta_F(\mu,r)$ in \eqref{eq:deltaEF}. By \eqref{def:deltap}, under pipelined transmission, the latency is determined by the maximum of fronthaul and edge latencies. Therefor, this first solution would leave open the possibility to increase the content multiplicity at the edge by sending more information through the fronthaul links to the ENs without increasing the system latency. This observation is leveraged by the proposed scheme. 

\begin{figure}[t!] 
  \centering
\includegraphics[width=0.48\columnwidth]{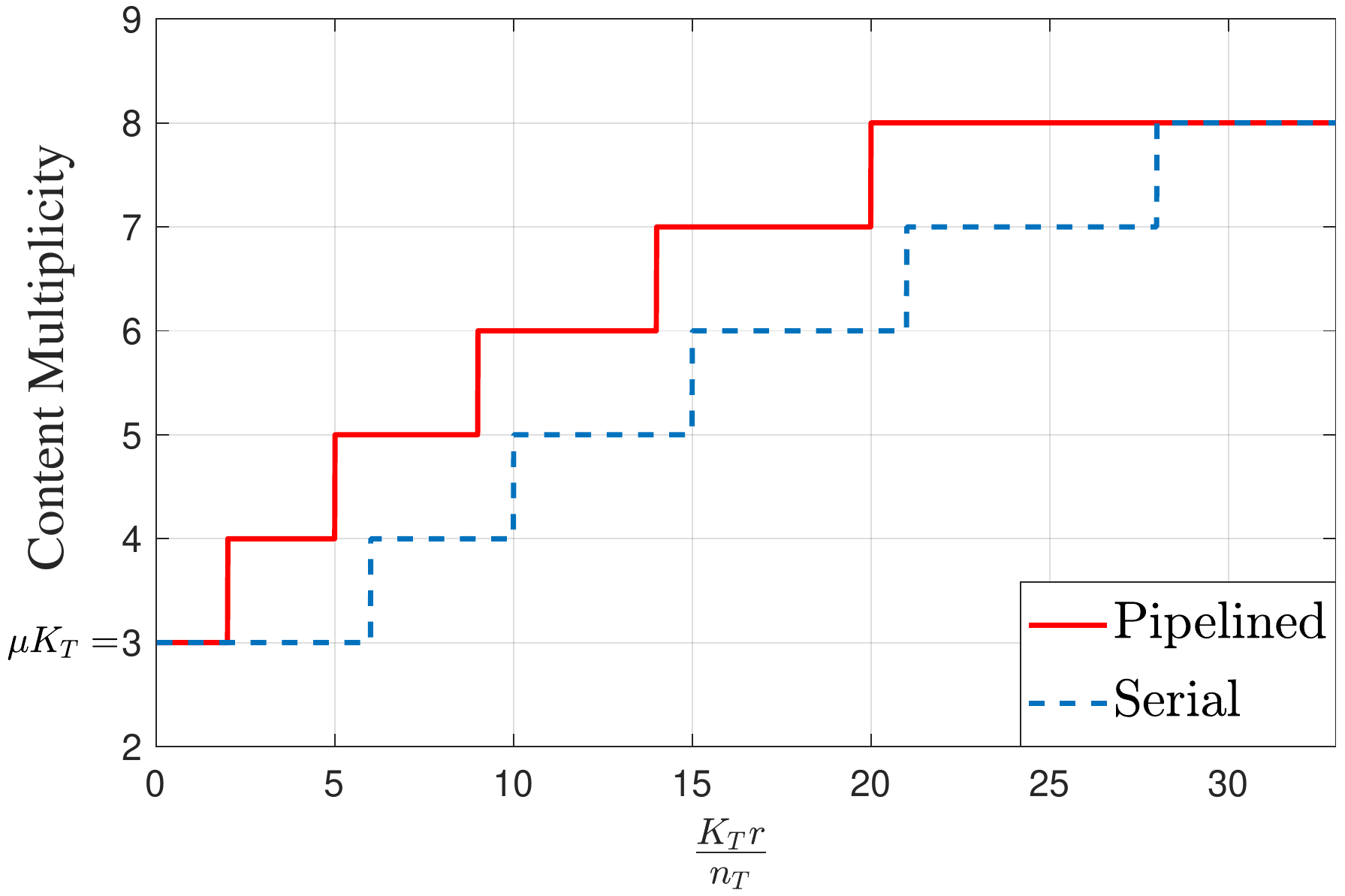}
\caption{Multiplicity $m_p(\mu,r)$ in \eqref{mp} under pipelined transmission and $m(\mu,r)$ in \eqref{m} under serial transmissions for $K_T=8, K_R=32, n_T=4$ and $\mu K_T=3$. }
\label{fig:mps}
\end{figure}

To start, we define a multiplicity $m_p(\mu,r)$ for the pipelined model, which is no less than the serial multiplicity $m(\mu,r)$ in \eqref{m}. This is done by first evaluating the multiplicity that ensures that fronthaul and edge transmission NDTs in \eqref{eq:deltaEF} and \eqref{eq:deltaEE1}, respectively, are equal, obtaining
\begin{align} \label{def:mps}
m_{eq}(\mu,r) = \frac{\mu K_T}{2}+\frac{\sqrt{(\mu K_T n_T)^2+4n_T K_T r}}{2 n_T}.
\end{align}
In order to account for the maximum multiplicity $m_{max}$ \eqref{mmax} and for the requirement that the multiplicity be an integer, we then define the multiplicity adopted by the proposed scheme as 
\begin{align} \label{mp}
m_p(\mu,r)= \left\{
\begin{array}{ll} 
\max\{\lfloor m_{eq}(\mu,r)  \rfloor,1 \}, &  \text{for}~ \mu K_T \leq m_{max}-\frac{K_T r}{m_{max} n_T} \\
       m_{max}, & \text{for}~  \mu K_T \geq m_{max}-\frac{K_T r}{m_{max} n_T}.
\end{array} 
\right.
\end{align}

As an example, the multiplicities $m(\mu,r)$ in \eqref{m} and $m_p(\mu,r)$ in \eqref{mp} under serial and pipelined transmissions, respectively, are plotted in Fig.~\ref{fig:mps} for $K_T=8, K_R=32, n_T=4$ and $\mu K_T=3$. Both multiplicities increase with $r$ from the multiplicity $\mu K_T$ that can be ensured by the edge cache resource only. The figure confirms that the proposed pipelined transmission scheme increases the multiplicity by leveraging simultaneous fronthaul and edge transmissions. The resulting achievable NDT is presented in the following proposition.

\begin{proposition} \label{pro:NDTp}  
For an F-RAN system with $n_T$ antennas, we have the upper bound $\delta_{p}^*(\mu,r)\leq \delta_{p,ach}(\mu,r)$ on the minimum NDT for pipelined fronthaul-edge transmission, where  
\begin{align} \label{achp}
\delta_{p,ach}(\mu,r)= \left\{
\begin{array}{ll} 
\max\big\{\frac{K_R(1-\mu K_T)}{K_T r}, 1 \big\}, &  \text{for}~ \mu K_T \leq \big(1-\frac{K_T r}{ n_T}\big)^+ \\
       \text{l.c.e.}\big\{ \max\big\{ \frac{K_R}{m_p(\mu,r)n_T}, 1 \}\big\}, & \text{for}~  \mu K_T \geq \big(1-\frac{K_T r}{ n_T}\big)^+
\end{array} 
\right.
\end{align} 
\end{proposition}
\begin{proof}
As discussed, the proposed scheme leverages block-Markov delivery and uses the multiplicity \eqref{mp}. As a result, the NDT is given by \eqref{def:deltap} with $\delta_F$ in \eqref{eq:deltaEF} and $\delta_E$ in \eqref{eq:deltaEE1}, where the multiplicity is in \eqref{mp}. Note that, with this scheme, for the small cache regime of $\mu K_T \leq (1-K_T r /n_T)^+$, by \eqref{mp}, we have the multiplicity $m_p(\mu,r)=1$ for each block, and the system performance is dominated by the fronthaul NDT $\delta_F$ in \eqref{eq:deltaEF}. Instead, for $\mu K_T \geq (1-K_T r /n_T)^+$, when the multiplicity $m_{p}(\mu,r)$ is an integer, i.e., when $m_{p}(\mu,r)=m_{eq}(\mu,r)$, the fronthaul and edge NDTs in \eqref{eq:deltaEF} and \eqref{eq:deltaEE1} are equal. When $m_{eq}(\mu,r)$ is not an integer, time sharing is performed between the two integer multiplicities $\lfloor m_{eq}(\mu,r)  \rfloor$ and $\lceil m_{eq}(\mu,r)\rceil$. 
\end{proof}

\begin{figure}[t!] 
  \centering
\includegraphics[width=0.5\columnwidth]{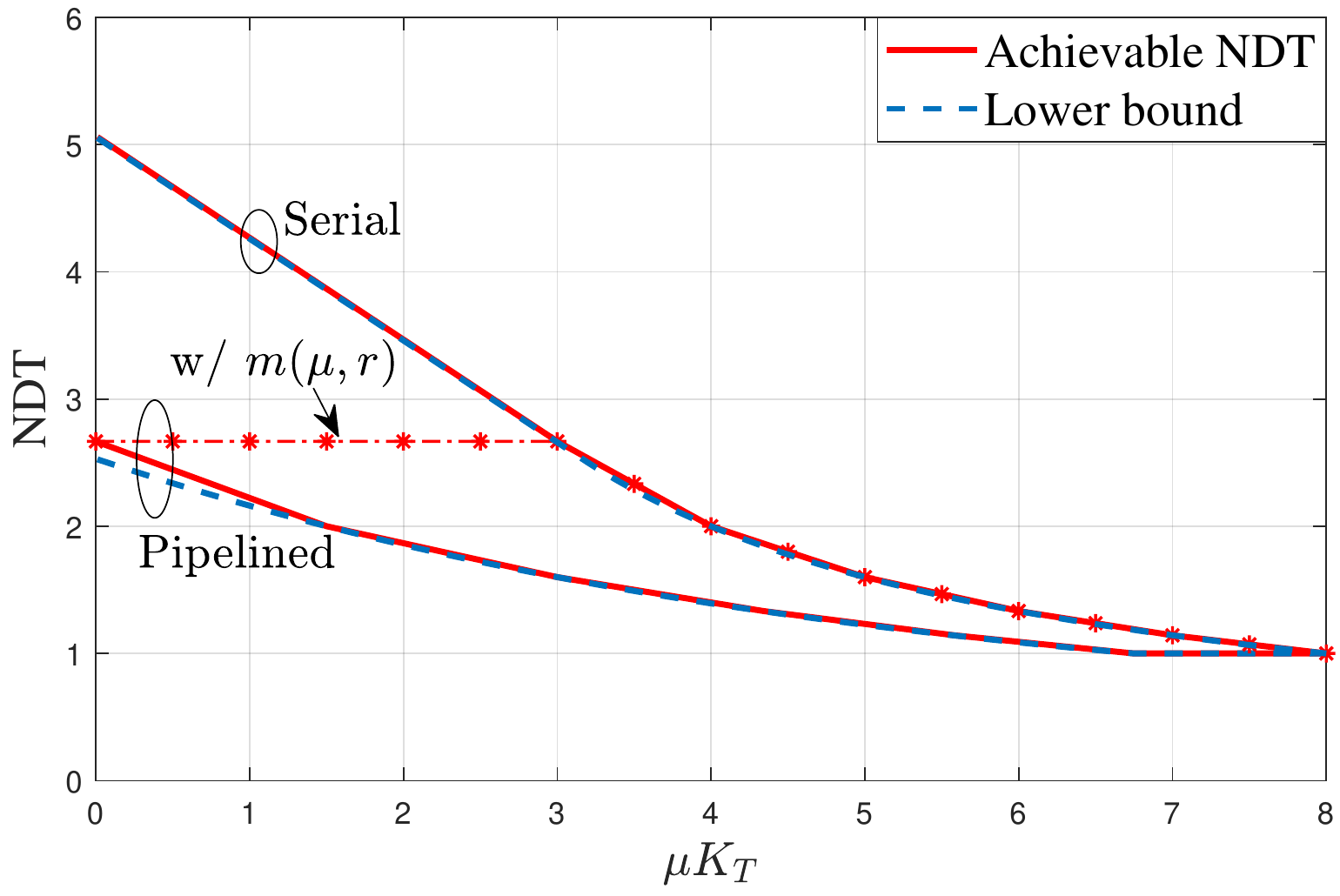}
\caption{Achievable NDTs under serial and pipelined transmissions with different multiplicity choices for $K_T=8, K_R=32, n_T=4$ and $r=5$. }
\label{fig:NDTcps}
\end{figure}

A comparison of the achievable NDTs under serial and pipelined transmissions for the same system parameters as in Fig.~\ref{fig:mps} can be found in Fig.~\ref{fig:NDTcps}. Beside the achievable NDTs in \eqref{ach:smallr} and \eqref{achp}, we also plot for reference the NDT of the pipelined policy obtained by using the same multiplicity $m(\mu,r)$ in \eqref{m} of the serial policy. Pipelining is seen to bring a non-negative reduction in NDT as compared to the serial policy even when the multiplicity is not optimized due to the possibility to use fronthaul and edge transmissions simultaneously. However, as discussed, the NDT performance with multiplicity $m(\mu,r)$ is dominated by edge transmission. As a result, when $\mu K_T \geq m(r)=3$, this policy does not provide NDT gains since, in this regime, only edge resources are used, i.e., $\delta_{F}=0$. In contrast, the proposed policy, with the multiplicity $m_p(\mu,r)$ in Fig.~\ref{fig:mps} can improve NDT for all values of $\mu \in [0,1]$.

Pipelined and serial NDTs \eqref{achp} and \eqref{ach:smallr} are further compared in Fig.~\ref{fig:NDTps} as a function of $\mu$ for two different values of $r$ for $K_T=8, K_R=32$ and $n_T=4$. For small values of $r$, such as $r=10^{-3}$ in the figure, the benefits from pipelining transmission is limited due to the bottleneck posed by the low fronthaul rate. Instead, for larger values of $r$, the reduction in NDT is evident. In particular, with pipelined delivery, when $\mu K_T \geq m_{max}-K_T r/(m_{max}n_T)=6.75$, the maximum multiplicity $m_{max}$ can be supported, yielding the ideal NDT of one. In contrast, for the serial case, the ideal NDT of one can be achieved only with full caching, i.e., $\mu=1$. Finally, as $\mu K_T$ increases, edge transmission becomes more efficient, reducing the contribution from fronthaul transmission and yielding a smaller improvement due to pipelined transmission.

\begin{figure}[t!] 
  \centering
\includegraphics[width=0.5\columnwidth]{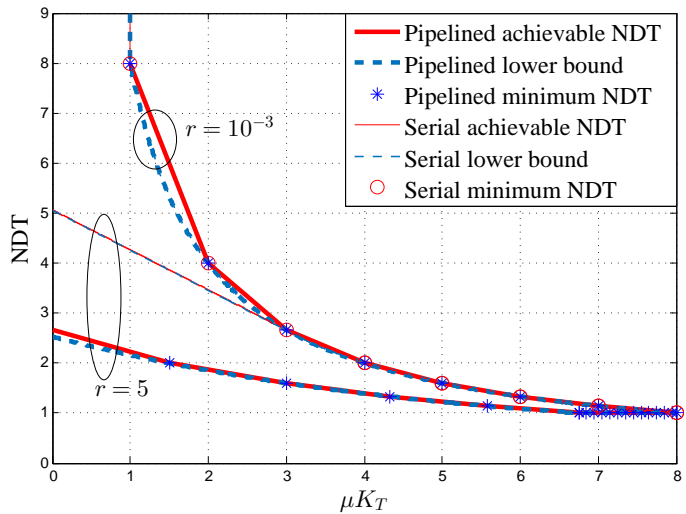}
\caption{Achievable NDT $\delta_{p,ach}(\mu,r)$ in Proposition~\ref{pro:NDTp} and lower bound $\delta_{p,lb}(\mu,r)$ Proposition~\ref{pro:NDTpmin} versus $\mu$ for different values of $r$, with $K_T=8, K_R=32$, and $n_T=4$.}
\label{fig:NDTps}
\end{figure}

\begin{remark} \label{remark2}
The proposed scheme is able to provide the discussed latency gains by leveraging 
\emph{cloud and edge-aided transmission} for any value of $\mu$ and $r$ except for the extreme case $\mu K_T \geq m_{max}$, where the maximum multiplicity $m_{max}$ can be ensured by caching only. We emphasize that this stands in stark contrast to the serial case in which, as observed in Remark \ref{remark}, when $\mu K_T \geq m(r)$, the overhead due to fronthaul transmission becomes excessive and edge transmission alone is preferable. 
\end{remark}

\subsection{Lower Bound on the Minimum NDT}
A lower bound on the minimum NDT $\delta_{p}^*(\mu,r)$ is presented in the following proposition, where we define the function $m_p^*(\mu,r)$ as
\begin{align} \label{mpins}
m_p^*(\mu,r)= \left\{
\begin{array}{ll} 
\max\{m_{eq}(\mu,r),1 \}, &  \text{for}~ \mu K_T \leq m_{max}-\frac{K_T r}{m_{max} n_T} \\
       m_{max}, & \text{for}~  \mu K_T \geq m_{max}-\frac{K_T r}{m_{max} n_T},
\end{array} 
\right.
\end{align}
with $m_{eq}(\mu,r)$ as in \eqref{def:mps}.
\begin{proposition} \label{pro:NDTpmin} 
For an F-RAN system with $n_T$ antennas, the minimum NDT $\delta_{p}^*(\mu,r)$ for pipelined fronthaul-edge transmission is lower bounded as 
%\begin{align} \label{mmin}
%\delta_{p}^*(\mu,r)\geq \delta_{p,lb}(\mu,r)= \max\Big\{\frac{K_R(1-\mu K_T)}{K_T r}, \frac{K_R}{m_p^*(\mu,r)n_T},1\Big\}.
%\end{align}
\begin{align} \label{deltamin}
\delta_{p}^*(\mu,r)\geq \delta_{p,lb}(\mu,r)= \left\{
\begin{array}{ll} 
\max\{\frac{K_R(1-\mu K_T)}{K_T r}, 1 \}, &  \text{for}~ \mu K_T \leq (1-\frac{K_T r}{ n_T})^+ \\
       \max\big\{ \frac{K_R}{m_p^*(\mu,r)n_T},1\big\}, & \text{for}~  \mu K_T \geq (1-\frac{K_T r}{ n_T})^+.
\end{array} 
\right.
\end{align}
\end{proposition}
\begin{proof}
The proof is presented in Appendix~\ref{sec:proofboundp}.
\end{proof}

\subsection{Minimum NDT}
The minimum NDT is characterized in the following Proposition for the regime of low cache and fronthaul capacities, i.e., with $\mu K_T \leq 1-K_T r/n_T$, and for any set-up with $\mu K_T=i-K_T r/(i n_T)$, $i\in [m_{max}]$ or with large caches, i.e., with $\mu K_T\geq m_{max}-K_T r/(m_{max} n_T)$. These conditions are akin to those identified in Proposition~\ref{pro3} for the serial case, with the different definitions being due to the use of fronthaul resources for all system parameters. 

\begin{proposition} \label{pro:optimal}
For an F-RAN system with $N_T$ antennas, the minimum NDT $\delta_{p}^*(\mu,r)$ for pipelined fronthaul-edge transmission is given as  
\begin{align} \label{eq:min}
\delta_{p}^*(\mu,r)= \left\{
\begin{array}{ll} 
\max\big\{\frac{K_R(1-\mu K_T)}{K_T r}, 1\big \}, &  \text{for}~ \mu K_T \leq 1-\frac{K_T r}{ n_T} \\
      \max\big\{ \frac{K_R}{m_p^*(\mu,r)n_T},1\big\}, & \text{for}~ \mu K_T \in \{i-\frac{K_T r}{i n_T}\}_{i=1}^{m_{max}}\cup (m_{max}-\frac{K_T r}{m_{max} n_T}, K_T].
\end{array} 
\right.
\end{align} 
\end{proposition}

\begin{proof}
The result can be obtained by the direct comparison of the bounds in Proposition~\ref{pro:NDTp} and Proposition~\ref{pro:NDTpmin}.
\end{proof}
\vspace{3pt}

The optimality regimes are illustrated in Fig.~\ref{fig:NDTps}. We emphasize that, in a manner similar to the serial case, for integer values of the multiplicity $m_p(\mu,r)=i$, the optimal NDT is achieved. The difference is that this multiplicity here is obtained with a cache size $\mu K_T=i-K_T r/ (i n_T)$, where $K_T r/ (i n_T)$ is the contribution to the multiplicity due to fronthaul transmission. 

The general multiplicative gap between the performance of the proposed scheme and the minimum NDT is stated in the following proposition.

\begin{proposition} \label{pro:gapp}
For a general F-RAN system with $n_T$ antennas at each EN, and any value of $\mu \geq 0$ and $r \geq 0$, we have the inequality 
\begin{align}
\frac{\delta_{p,ach}(\mu,r)}{\delta_{p}^*(\mu,r)} \leq 2. 
\end{align}
\end{proposition}
\begin{proof}
The proof is presented in Appendix~\ref{sec:proofgapp}.  
\end{proof}

\section{Conclusions} \label{sec:con}
In fog-aided cellular systems, fronthaul resources enable a cloud processor with access to the content library to communicate uncached contents to the edge nodes. This information is not only necessary to enable content delivery when the overall system's capacity is insufficient, but it can also facilitate cooperative interference management. In this paper, we have studied the resulting optimal trade-off between fronthaul latency overhead and overall delivery latency from an information-theoretic viewpoint under the assumption of multi-antenna edge nodes, uncoded caching and fronthaul and one-shot linear precoding on the wireless edge channel. The minimum delivery latency was investigated in the high-SNR regime under both serial and pipelined transmission models. The main results of this F-RAN model are the characterizations, within small multiplicative factors, of the minimum high-SNR latency as a function of system parameters such as fronthaul capacity, edge cache capacity and number of per-edge node antennas. Extensive numerical results have been provided to demonstrate the usefulness of the derived information-theoretic characterizations in understanding the interplay and relative roles of edge and cloud resources on the performance of fog-aided networks. We have also commented on the impact of the practical assumptions made here as compared to the unconstrained delivery strategies studied in \cite{STS:17}.

The information-theoretic characterizations derived in this work leave open a number of research questions. A first line of work that has recently been partially addressed in \cite{MGL:17,KACSP:18,CTS:17} concerns the effect of imperfect or no CSI on the optimal design of caching and delivery techniques. A second, related, issue is the study of optimal edge caching techniques under partial connectivity \cite{CTS:17,BGL:15}. Third, the analysis can be extended to yield insights into the performance of online caching strategies by following \cite{ASST:17}. Lastly, using ideas from \cite{NMA:17}, it would be interesting to generalize the results of this work to a set-up that includes also caching at the users.

\begin{appendices}
\section{Proof of Proposition~\ref{pro:NDT}} \label{sec:proof}

As discussed in Section~\ref{sec:edgeft}, for a desired multiplicity $m\geq \lfloor \mu K_T \rfloor$, we distinguish the cases $m=\lfloor \mu K_T \rfloor$ and $m>\lfloor \mu K_T \rfloor$. In the first case, edge-only transmission is used, and we adopt the same cache-aided delivery strategy described in Section~\ref{sketchedge}, yielding the edge NDT in \eqref{ach:edge}. 

%As describes in the model in Section~\ref{model}, with the packetization $F$ of a file, each phase of the whole communication, including caching, fronthauling and wireless delivery, can operate on a subset of resulting units from the packetization. 
%%We define $F$ as the number of packetization of a file such that caching, fronthauling and wireless delivery can operate with a subset of units from the packetization. 
%In Section~\ref{sec:edge} and Section~\ref{sec:fronthaul} with each having access to one resource only, we have that $F=F_CF_D$ for both cases with $F_C$ and $F_D$ defined in \eqref{def:subedgec} and \eqref{def:subedged}, respectively. Instead, for the general F-RAN including the above two specific cases, the packetization $F$ is adapted to the corresponding strategies.

In contrast, for the case $m>\lfloor \mu K_T \rfloor$, cloud and edge-aided transmission is used, and the multiplicity $m$ is obtained using both caching and fronthaul transmission during the delivery phase. Generalizing Example 4, we first divide each content into $F_C$ parts $\{W_{ni}\}_{i=1}^{F_C}$, with $F_C$ in \eqref{def:subedgec}, and then we further divide each part into 
%Comparing to the above case, the edge-only and fronthaul-only caching, this comes with a distinct packetization to distinguish two resources. To elaborate, the number of parts required to ensure that subsets of $m$ ENs receive the same part is again $F_C$ in \eqref{def:subedgec}, whereas the number of packets generated from each part is defined as 
\begin{align} \label{FDP}
F_S=\text{l.c.m.} (F_D,m)
\end{align}
packets $\{W_{nij}\}_{j=1}^{F_S}$, with $F_D$ defined in \eqref{def:subedged}. As a result, the overall number of packets is
\begin{align} \label{FNew}
F=F_C F_S. 
\end{align}
By \eqref{FNew}, we have the inequalities $K_T\leq F\leq K_T K_R m$.

In the caching phase, the $F_S$ packets are arbitrarily divided into two disjoint subsets $W_{ni}^1$ and $W_{ni}^2$, where subset $W_{ni}^1$ contains an integer number $F_S\lfloor \mu K_T \rfloor/m$ of packets and subset $W_{ni}^2$ contains the rest. During the caching phase, all the packets in subsets $\{W_{ni}^1\}_{n=1}^{N}$ are cached at all $m$ EN in cluster $\mathcal{K}_{Ti}$ by following \eqref{def:clusteredge}, while the packets in subsets $\{W_{ni}^2\}_{n\in[N],i\in[F_C]}$ are left uncached.

During the delivery phase, for a demand vector $\dv$, the uncached $K_RF_S(1-\lfloor \mu K_T \rfloor/m)$ packets in subsets $\{W_{d_k i}^2\}_{k\in[K_R]}$, with $i\in[F_C]$, are sent to all $m$ ENs in cluster $\mathcal{K}_{Ti}$ as in \eqref{def:clusteredge}. Hence, each EN receives $K_R F(m-\lfloor \mu K_T \rfloor)/K_T$ packets on the fronthaul, yielding the fronthaul NDT $\delta_F(m)=|\mathcal{F}_i|/Fr=K_R (m-\lfloor \mu K_T \rfloor)/(K_T r)$. As a result of fronthaul transmission, for each file $W_{d_k}, d_k\in\dv$, the ENs in each cluster $\mathcal{K}_{Ti}$ share all $F_S$ packets in subsets $\{W_{d_k i}^1\}$ and $\{W_{d_k i}^2\}$. These ENs can hence transmit cooperatively to the $B_D$ groups $\{\mathcal{K}_{Rj}\}_{j=1}^{B_D}$ of $u(m)=mn_T$ users defined in \eqref{def:groupedge} by using $(F_S/F_D)B_D$ blocks. Hence, the total number of blocks is
\begin{align} \label{def:block}
B=\frac{F_S}{F_D}B_D F_C,
\end{align}
yielding the edge NDT in \eqref{def:DE}, i.e., $\delta_E(m)=B/F=B_D/F_D=K_R/(mn_T)$.

%As a result, the ENs in each cluster $\mathcal{K}_{Ti}$ share all the $K_RF_S$ packets in subsets $\{W_{d_k i}^1\}_{k\in[K_R]}$ and $\{W_{d_k i}^2\}_{k\in[K_R]}$, which can be sent to $B_D$ groups of $u(m)=mn_T$ users defined in \eqref{def:groupedge} within $K_RF_S/u(m)$ blocks, the resulting number of blocks is 
 %\begin{align} \label{def:b\textit{}lock}
%B=\frac{F_C K_R F_S}{mn_T},
%\end{align}
%yielding the edge NDT in \eqref{def:DE} $\delta_E(m)=B/F=K_R/(mn_T)$.

As a result, for a given multiplicity $m$, the overall NDT is given as $\delta(m)=\delta_E(m)+\delta_F(m)$ with integer $m\in[\lfloor \mu K_T \rfloor, m_{max}]$. To minimize the NDT $\delta(m)$, we define the function $\delta(x)$ as
\begin{align}
\delta(x)=\frac{K_R (x-\lfloor \mu K_T \rfloor)}{K_T r}+\frac{K_R}{x n_T},
\end{align}
where $x\in[\lfloor \mu K_T \rfloor,m_{max}]$ is a variable obtained by relaxing the integer constraints over $m$. Function $\delta(x)$ is convex within its domain, and has only one stationary point $x_0=\sqrt{K_Tr/n_T}$. Hence, it reaches the minimum at point $x=x_0$ for $\lfloor \mu K_T \rfloor\leq x_0$, and point $x=\lfloor \mu K_T \rfloor$ for $\lfloor \mu K_T \rfloor\geq x_0$. While in the latter case the solution is an integer, in the former case, the optimal solution is given as $\lfloor x_0 \rfloor$ if $\delta(\lfloor x_0 \rfloor)<\delta(\lceil x_0 \rceil)$ or $\lceil x_0 \rceil$ if $\delta(\lfloor x_0 \rfloor)>\delta(\lceil x_0 \rceil)$. Here, in order to simplify the expressions, we select $[x_0]$ when $\lfloor \mu K_T \rfloor \leq m(r)$.

\section{Proof for Proposition 2} \label{converse}

The proof follows [5, Section 5] with the important caveats that here we need to additionally consider the delivery latency due to fronthaul transmission, as well as the extension to the general case $n_T\geq 1$. To start, we consider an arbitrary split of each file into $2^{K_T}-1$ parts, such that each part $W_{n\tau}$, indexed by a subset $\tau \subseteq [K_T]$, contains an integer number of packets, including possibly no packets. We recall that each packet contains $L/F$ bits. Part $W_{n\tau}$ is available at the ENs in the subset $\tau$, either from the edge caches or from the cloud after fronthaul transmission. Note that this partition comes with no loss of generality, since each packet $W_{nf}$ is available at all EN $i$ such that $W_{nf} \in \mathcal{C}_i \cup \mathcal{F}_i$ (see definitions in Section~\ref{model}).

To distinguish between the contributions of cache and fronthaul resources, we use $c_{n\tau}$ to denote the number of cached packets from file $W_n$ at the ENs in subset $\tau$; while $f_{n\tau}(\dv)$ is the number of packets of file $W_n$ sent on the fronthaul links of all ENs in subset $\tau$ for a given demand vector $\dv$. Hence, part $W_{n\tau}$ has $a_{n\tau}=c_{n\tau}+f_{n\tau}(\dv)$ packets in total. The variables $\{c_{n\tau}\}$ and $\{f_{n\tau}(\dv)\}$, for all $ n \in[N], \tau \subseteq [K_T]$ and vectors $\dv$, fully specify the operation of the cache strategy $\mathcal{C}_i$ and fronthaul policy $\mathcal{F}_i$ defined in Section~\ref{model}.

Minimizing the NDT with respect to the caching strategy $\{c_{n\tau}\}_{n\in[N],\tau \subseteq \mathcal{T}}$ and fronthaul policy $\{f_{n\tau}(\dv)\}_{n\in\dv,\tau \subseteq \mathcal{T}}$ for all vectors $\dv$ yields the following integer problem
\begin{subequations} \label{pro:1}
\begin{align}
&\underset{\{c_{n\tau}\}, \{f_{n\tau}(\dv)\}}{\text{minimize}} \max_{\dv} ~\delta^*_E\big(\{c_{n\tau}\}, \{f_{n\tau}(\dv)\},\dv \big) + \delta_F^*(\dv)  \label{con:0} \\
&\text{s.t.}\sum_{i=1}^{K_T} \sum_{\substack {{\tau \subseteq [K_T]:} \\ {|\tau|=i}}} (c_{n\tau}+f_{n\tau}(\dv))=F, \forall n \in \dv , \forall \dv  \label{con:1}\\
&~~~ \sum_{n=1}^{N} \sum_{\substack {{\tau \subseteq [K_T]:} \\ {i\in \tau}}} c_{n\tau} \leq \mu FN , \forall i\in [K_T] \label{con:2}\\
&~~~\frac{1}{Fr} \sum_{n\in \dv} \sum_{\substack {{\tau \subseteq [K_T]:} \\ {i\in \tau}}}f_{n\tau}(\dv) \leq \delta_F^*(\dv) , \forall i \in [K_T], \forall  \dv \label{con:3} \\
&~~~c_{n\tau}  \geq 0, f_{n\tau}(\dv) \geq 0  \label{con:4} \\
&~~~0\leq \delta_F^*(\dv) \leq \delta_{Fmax}, \label{con:5}
\end{align}
\end{subequations}
where $\delta^*_E\big(\{c_{n\tau}\}, \{f_{n\tau}(\dv)\},\dv \big)$ is the minimum edge NDT \eqref{def:DE} for given cache and fronthaul policies when the request vector is $\dv$. In \eqref{con:1}, the equality constraints enforce that all $F$ packets of each requested file are available collectively at the ENs after the fronthaul transmission; inequalities \eqref{con:2} come from the fact that the size of the cache content $\mathcal{C}_i$ of each EN $i$, which is given as $\sum_{n=1}^{N} \sum_{\tau \subseteq [K_T]:i\in \tau} c_{n\tau}$, is constrained by the cache capacity $\mu FN$ (see \eqref{cap:cache}); inequalities \eqref{con:3} follow from the definition of fronthaul NDT \eqref{def:DF}, since the left-hand side is the number of packets sent to EN $i$ on the fronthaul for request vector $\dv$;
and inequalities \eqref{con:5} impose that the fronthaul NDT is no larger than 
\begin{align}
\delta_{Fmax} = \frac{K_R (m_{max}-\mu K_T)^+}{K_T r}.
\end{align}
This is because, as discussed in Section \ref{sec:sketch}, the multiplicity of the requested files can be upper bounded without loss of generality by $m_{max}$, and the maximum overall number of bits that are needed from the cloud to ensure this multiplicity is given as $K_R(m_{max}-\mu K_T)^+L$ bits.

The optimum value of optimization problem \eqref{pro:1} is lower bounded by substituting the maximum over all the request vector $\dv$ with an average. In particular, since the number of ways to request all the $K_R$ distinct files out of $N$ library files is $\pi(N,K_R) = N!/(N-K_R)!$, the lower-bounding problem can be written as  
\begin{subequations} \label{pro:2}
\begin{align}
&\underset{\{c_{n\tau}\}, \{f_{n\tau}(\dv)\}}{\text{minimize}} \frac{1}{\pi(N,K_R)}\sum_{\dv} \delta^*_E\big(\{c_{n\tau}\}, \{f_{n\tau}(\dv)\},\dv \big) + \delta_F^*(\dv) \label{con:21} \\
&\text{s.t.}~\eqref{con:1}-\eqref{con:5}.
\end{align}
\end{subequations} 

 We now obtain a lower bound on the optimal value of problem \eqref{pro:2} and hence also of problem \eqref{pro:1}. To this end, we first bound the minimum edge NDT $\delta^*_E\big(\{c_{n\tau}\}, \{f_{n\tau}(\dv)\},\dv \big)$ in \eqref{con:21} by studying the number of packets that can be served in each block as a function of the availability of files at the ENs. 

\begin{lemma} \label{lem:number}
Consider a single edge transmission block $b$ in which a set $\{W_{n_lf_l}\}_{l=1}^{L}$ of $L$ packets are sent to $L$ distinct users in set $\mathcal{R}(b)\subseteq[K_R]$. In order for each user in $\mathcal{R}(b)$ to be able to decode the desired packet without interference at the end of the block, the number $L$ of packets must be upper bounded as 
\begin{align}
L \leq \min_{l \in [L]} |\tau_l| n_T, \label{lemma}
\end{align}
where for any packet $W_{n_lf_l}$, $ \tau_l$ denotes the subset of ENs that have access to it, either as part of the pre-stored contents at the EN's cache or of the fronthaul received signals, i.e., $ W_{n_lf_l}\in \{\mathcal{C}_i\cup \mathcal{F}_i\}_{i\in \tau_l}$.
\end{lemma}

\begin{proof}
The proof follows from [5, Lemma 3] with the following differences. For a block $b$, each EN $i$ sends
\begin{align}
\xv_i(b) = \sum_{l:i \in \tau_l} \vv_{i n_l f_l}(b) s_{n_l f_l}(b),
\end{align}
and the received signal at user $k \in \mathcal{R}(b)$, is given as 
\begin{align}
y_k(b)&= \sum_{i=1}^{K_T} \hv_{ki}^T(b) \xv_i(b)+ z_k(b) \\
       & = \sum_{i=1}^{K_T} \hv_{ki}^T(b)\sum_{l: i \in \tau_l} \vv_{i n_l f_l}(b) s_{n_l f_l}(b) + z_k(b) \\
			& = \sum_{l=1}^{L}\sum_{i \in \tau_l} \hv_{ki}^T(b) \vv_{i n_l f_l}(b) s_{n_l f_l}(b) + z_k(b). \label{conversion}
\end{align}
From \eqref{conversion}, the channel can be considered as a multi-antenna broadcast channel with $L$ transmitters, each having $|\tau_l| n_T$ antennas, that are connected to $L$ single-antenna users. By following the same steps as in [5, Eq. (28)-(36)] the proof is completed.
\end{proof}

Each subset $\tau$ of ENs needs to deliver parts $\{W_{n,\tau}\}$, $n \in \dv$, which consists of a total of $\sum_{j}(c_{d_j\tau}+f_{d_j\tau})$ packets. From Lemma~\ref{lem:number}, the number of necessary blocks is at least $\sum_{j}(c_{d_j\tau}+f_{d_j\tau})/(|\tau| n_T)$. 
%First, in order to guarantee that each requested file $W_n$, with $n \in \dv$, is available at the ENs, we must have the equality  $\sum_{\tau \in [K_T]} a_{n,\tau} = F$. So that each requested part $W_{n\tau}$ is available at some ENs, either from the edge caches or from the cloud. To distinguish between the two resources, we use $c_{n\tau}$ to denote the number of packets from the former, and $f_{n\tau}(\dv)$ from the latter. Note that we have $a_{n\tau}=c_{n\tau}+f_{n\tau}(\dv)$. 
By summing over all subsets $\tau$ and applying \eqref{def:DE}, the minimum edge NDT $\delta^*_E (\{c_{n\tau}\}, \{f_{n\tau}(\dv)\}, \dv)$ can be lower bounded as 
\begin{align} \label{eq:deltaE}
\delta^*_E\big(\{c_{n\tau}\}, \{f_{n\tau}(\dv)\}, \dv\big) \geq \frac{1}{F}\sum_{i=1}^{K_T} \sum_{j=1}^{K_R} \sum_{\substack {{\tau \subseteq [K_T]}: \\ {|\tau|=i}}} \frac{c_{d_j,\tau}+f_{d_j,\tau}}{i n_T}.
\end{align}
This bound is instrumental in proving the following lemma, which completes the proof upon combination with the trivial lower bound $K_R/\min\{K_Tn_T, K_R\}$ on the edge NDT.

\begin{lemma} \label{lem:fmin}
The optimal value of the problem \eqref{pro:2} is lower bounded by 
\begin{align} \label{eq:fmin}
f_{min}=\left\{
\begin{array}{ll}
       \frac{K_R(m^*(r)-\mu K_T)}{K_T r}+\frac{K_R}{m^*(r) n_T},& \mu K_T < m^*(r) \\ [1ex]
     \frac{K_R}{\mu K_T n_T}, &  \mu K_T \geq m^*(r), \\ [1ex]
\end{array} 
\right.
\end{align}
where $m^*(r)$ is defined in \eqref{mmins}.
\end{lemma}
\begin{proof}
The proof is presented in Appendix \ref{lemma2}.
\end{proof}

\section{Proof of Lemma \ref{lem:fmin}} \label{lemma2}

We lower bound the two terms in \eqref{con:21} separately by starting with the minimum average edge NDT
%We start by lower bounding the minimum average edge NDT in \eqref{con:21} as 
\begin{subequations} \label{low:de}
\begin{align} 
%&\bar\delta_E\Big(\{\mathcal{C},\mathcal{F}\}_{i=1}^{K_T},\dv\Big) \notag\\
&\frac{1}{\pi(N,K_R)}\sum_{\dv} \delta^*_E\big(\{c_{n\tau}\}, \{f_{n\tau}(\dv)\},\dv \big)  \\
& \stackrel{(a)}{\geq}  \frac{1}{F\pi(N,K_R)} \sum_{i=1}^{K_T}\frac{1}{i n_T}\Bigg[\sum_{\dv} \sum_{j=1}^{K_R} \sum_{\substack {{\tau \subseteq [K_T]:} \\ {|\tau|=i}}}\big( c_{d_j\tau}+f_{d_j\tau}(\dv)\big) \Bigg]  \\
&\stackrel{(b)}{=}\frac{1}{F\pi(N,K_R)} \sum_{i=1}^{K_T}\frac{1}{i n_T}\Bigg[ K_R \!\!\! \sum_{\substack {{\tau \subseteq [K_T]:} \\ {|\tau|=i}}} \!\!\! \pi(N-1,K_R-1) \!\! \sum_{n=1}^{N} (c_{n\tau}+\tilde{f}_{n\tau} )\Bigg]  \\
%&~~~+ \sum_{\substack {{\tau \subseteq [K_T]:} \\ {|\tau|=i}}} \sum_{\dv} \sum_{n\in \dv} f_{n\tau}(\dv) \Big]\notag \\
&=\frac{K_R}{NF} \sum_{i=1}^{K_T}\frac{1}{i n_T}\Bigg[ \sum_{\substack {{\tau \subseteq [K_T]:} \\ {|\tau|=i}}} \sum_{n=1}^{N} (c_{n\tau}+\tilde{f}_{n\tau} )\Bigg]   \\
&\stackrel{(c)}{=}\frac{K_R}{N F n_T}  \sum_{i=1}^{K_T}\frac{1}{i} b_i  \\
&\stackrel{(d)}{\geq} \frac{K_R}{N F n_T}  \frac{(\sum_{i=1}^{K_T} b_i)^2}{\sum_{i=1}^{K_T} i b_i}, \label{eq:deltas} 
\end{align}
\end{subequations} 
where inequality (a) follows from inequality \eqref{eq:deltaE}; equality (b) holds because, for any library file $W_n$, the number of different request vectors that include file $W_n$ is $K_R \pi(N-1,K_R-1)$, i.e., $\sum_{\dv} \sum_{j=1}^{K_R}  W_{d_j\tau}=K_R  \pi(N-1,K_R-1) \sum_{n=1}^{N} W_{n\tau}$, and hence we have 
\begin{subequations} \label{eq:cf}
\begin{align} 
\sum_{\dv} \sum_{j=1}^{K_R}  c_{d_j\tau}&=K_R  \pi(N-1,K_R-1) \sum_{n=1}^{N} c_{n\tau} \\
\text{and}~~ \sum_{\dv} \sum_{j=1}^{K_R} f_{d_j\tau}(\dv) &=\sum_{\dv} \sum_{n\in \dv} f_{n\tau}(\dv)=K_R\pi(N-1,K_R-1)\sum_{n=1}^{N}\tilde{f}_{n\tau}, \label{def:tildef}
\end{align} 
\end{subequations} 
where $\tilde{f}_{n\tau}=\sum_{\dv:n\in \dv} f_{n\tau}(\dv)/(K_R\pi(N-1,K_R-1))$ represents the number of packets in part $W_{n\tau}$ for each user in $\tau$ received from the cloud; equality (c) follows the definition 
\begin{align}
b_i =  \sum_{\substack {{\tau \subseteq [K_T]:} \\ {|\tau|=i}}} \sum_{n=1}^{N} (c_{n\tau}+\tilde{f}_{n\tau} ); \label{def:bi}
\end{align}
and inequality (d) applies the Cauchy-Schwarz inequality $(\sum_{i=1}^{n} u_iv_i)^2 \leq (\sum_{i=1}^{n} u_i^2) (\sum_{i=1}^{n} v_i^2$) by setting $u_i=\sqrt{b_i/i}$ and $v_i=\sqrt{i b_i}$.

To compute the term $\sum_{i=1}^{K_T} b_i$ in \eqref{eq:deltas}, we impose the constraint \eqref{con:1}, obtaining  
\begin{subequations}  \label{eq:fileav}
\begin{align}
\pi(N,K_R) K_R F & \stackrel{(a)}{=}\sum_{\dv}\sum_{n\in\dv}\sum_{i=1}^{K_T}  \sum_{\substack {{\tau \subseteq [K_T]:} \\ {|\tau|=i}}} (c_{n\tau}+f_{n\tau}(\dv))  \\
& \stackrel{(b)}{=}K_R \pi(N-1,K_R-1) \sum_{i=1}^{K_T} \sum_{\substack {{\tau \subseteq [K_T]:} \\ {|\tau|=i}}}\sum_{n=1}^{N} (c_{n\tau}+\tilde{f}_{n\tau})  \\
& \stackrel{(c)}{=}K_R \pi(N-1,K_R-1) \sum_{i=1}^{K_T} b_i,
\end{align}
\end{subequations}
where equality (a) holds by summing up the constraints in \eqref{con:1} for all $\pi(N,K_R)$ request vectors and for all $K_R$ files in each vector $\dv$; and equalities (b) and (c) follow from the equalities in \eqref{eq:cf} and the definition of $b_i$ in \eqref{def:bi}, respectively. From \eqref{eq:fileav}, we have the equality
$\sum_{i=1}^{K_T} b_i=NF$.

We move on to lower bound the second term in \eqref{con:21}, i.e., the minimum fronthaul NDT $\delta_F^*(\dv)$. We start by bounding the size of the cached content. From \eqref{con:2}, we have 
\begin{subequations}
\begin{align} \label{cons:cache}
\mu FN K_T & \stackrel{(a)}{\geq}\sum_{i=1}^{K_T} \sum_{n=1}^{N} \sum_{\substack {{\tau \subseteq [K_T]:} \\ {i\in \tau}}} c_{n\tau}=   \sum_{n=1}^{N} \sum_{i=1}^{K_T} \sum_{\substack {{\tau \subseteq [K_T]:} \\ {i\in \tau}}} c_{n\tau} \notag \\
& \stackrel{(b)}{=} \sum_{n=1}^{N}  \sum_{i=1}^{K_T} i \sum_{\substack {{\tau \subseteq [K_T]:} \\ {|\tau|=i}}} c_{n\tau} =  \sum_{i=1}^{K_T}i  \sum_{\substack {{\tau \subseteq [K_T]:} \\ {|\tau|=i}}}\sum_{n=1}^{N}  c_{n\tau},
\end{align}
\end{subequations}
where inequality (a) holds by summing the inequalities in \eqref{con:2} for all the $K_T$ ENs; and equality (b) comes from the fact that the size of the cached content of a file $W_n$ across the ENs is given as $\sum_{i=1}^{K_T} \sum_{\tau: i\in \tau} c_{n\tau} = \sum_{i=1}^{K_T} i \sum_{\tau:|\tau|=i}c_{n\tau}$.

With the above inequality, the minimum fronthaul NDT can be bounded as 
\begin{subequations} \label{low:df}
\begin{align} \label{con:deltaf}
\frac{1}{\pi(N,K_R)}\sum_{\dv} \delta_F^*(\dv)  &\stackrel{(a)}{\geq}\frac{1}{\pi(N,K_R)} \sum_{\dv} \frac{1}{K_T}\sum_{i=1}^{K_T}\frac{1}{Fr} \sum_{n \in \dv} \sum_{\substack {{\tau \subseteq [K_T]:} \\ {i\in \tau}}}f_{n\tau}(\dv)   \\
                          &\stackrel{(b)}{=} \frac{1}{\pi(N,K_R)}\frac{1}{K_TFr}\sum_{\dv} \sum_{n \in \dv} \sum_{i=1}^{K_T} i \sum_{\substack {{\tau \subseteq [K_T]:} \\ {|\tau|=i}}}f_{n\tau}(\dv)   \\
													&\stackrel{(c)}{=} \frac{K_R}{N K_T Fr} \sum_{i=1}^{K_T} i \sum_{\substack {{\tau \subseteq [K_T]:} \\ {|\tau|=i}}} \sum_{n=1}^{N}  \tilde{f}_{n\tau} 	\\
													&\stackrel{(d)}{=} \frac{K_R}{N K_T Fr} \sum_{i=1}^{K_T} i \Bigg(b_i-\sum_{\substack {{\tau \subseteq [K_T]:} \\ {|\tau|=i}}}\sum_{n=1}^{N}  c_{n\tau} \Bigg)  \\
													&\stackrel{(e)}{\geq} \frac{K_R}{ K_T r} \Bigg(\frac{1}{NF}\sum_{i=1}^{K_T} i b_i-\mu K_T\Bigg),
\end{align}
\end{subequations}
where inequality (a) holds by averaging the constraints in \eqref{con:3}; equality (b) follows in a manner similar to equality (b) in \eqref{cons:cache}; equalities (c) and (d) follow the equality in \eqref{def:tildef} and the definition of $b_i$ in \eqref{def:bi}, respectively; and inequality (e) holds by using \eqref{cons:cache}.

Now we can bound the minimum NDT by using \eqref{eq:deltas}, \eqref{eq:fileav} and \eqref{con:deltaf} as 
\begin{subequations}
\begin{align}
%&\bar\delta_E\Big(\{\mathcal{C},\mathcal{F}\}_{i=1}^{K_T},\dv\Big) + \delta_F \\
&\frac{1}{\pi(N,K_R)} \sum_{\dv} \delta^*_E\big(\{c_{n\tau}\}, \{f_{n\tau}(\dv)\},\dv \big)+ \delta_F^*(\dv)   \\
%&\geq \frac{K_R}{N n_T F}  \sum_{i=1}^{K_T}\frac{1}{i} b_i+ \delta_F 
%\geq \frac{K_R}{N n_T F}  \frac{(\sum_{i=1}^{K_T} b_i)^2}{\sum_{i=1}^{K_T} i b_i}+ \delta_F(\dv)  \label{final1} \\
%&\geq \frac{K_R}{N n_T F}  \frac{(\sum_{i=1}^{K_T} b_i)^2}{\sum_{i=1}^{K_T} i b_i}+ \delta_F  \\
&\geq \frac{K_R}{N F n_T}  \frac{(NF)^2}{\sum_{i=1}^{K_T} i b_i}+ \frac{K_R}{ K_T r} \Bigg(\frac{1}{NF}\sum_{i=1}^{K_T} i b_i-\mu K_T\Bigg) \\
%&= \frac{K_R}{n_T} \frac{1}{x}+\frac{K_R}{K_T r}(x-\mu K_T) \label{final3} \\
&= \frac{K_R (x-\mu K_T)}{K_T r } + \frac{K_R}{n_T} \frac{1}{x}, \label{final:x}
\end{align}
\end{subequations}
where in the last step, we have defined the variable $x=\sum_{i=1}^{K_T} i b_i/(NF)$. Since, by \eqref{def:bi}, the expression $\sum_{i=1}^{K_T} i b_i$ is the overall number of packets of all library files that are available upon fronthaul transmission at subsets of ENs of any size $i$, the variable $x$ can be interpreted as the average multiplicity of each file at the ENs after fronthaul transmission. 

From \eqref{final:x}, we define the function 
\begin{align} \label{def:fx}
f(x)=  \frac{K_R (x-\mu K_T)}{K_T r } + \frac{K_R}{n_T} \frac{1}{x}.
\end{align} 
To complete the proof, we now minimize $f(x)$ in \eqref{def:fx} over $x$. To this end, we first focus on defining the domain of $x$. From \eqref{con:5} and \eqref{con:deltaf}, we have the bounds $K_R(x-\mu K_T)/(K_T r) \leq  \delta_F(\dv)\leq\delta_{Fmax}$, yielding the upper bound $x\leq (m_{max}-\mu K)^++\mu K_T=\max\{m_{max},\mu K_T\}$. We also have the inequality $x \geq \mu K_T$ due to the bound $\delta_F^*(\dv)\geq 0$.
Furthermore, from \eqref{eq:fileav}, we have the inequality $\sum_{i=1}^{K_T} b_i/NF \geq 1$, yielding $x=\sum_{i=1}^{K_T} i b_i/NF \geq 1$. In summary, variable $x$ needs to lie in the interval $\max\{1,\mu K_T\}=x_{min}\leq x\leq x_{max}=\max\{m_{max},\mu K_T\}$. We then turn to minimizing the function $f(x)$ in the interval $x\in [x_{min}, x_{max}]$. Function $f(x)$ is convex for $x>0$, and the only stationary point is $x=\sqrt{K_T r/n_T}$, i.e., $f'\big(\sqrt{K_T r/n_T}\big)=0$. Therefore, the desired minimum $f_{min}$ is given as 
\begin{align} 
f_{min}= \left\{
\begin{array}{ll} 
f\big(\sqrt{K_T r/n_T}\big), &  \text{if}~ x_{min}\leq \sqrt{K_T r/n_T} \leq x_{max}\\
       \min\{f(x_{min}), f(x_{max})\}, & \text{otherwise},
\end{array} 
\right.
\end{align}
which is as reported in \eqref{eq:fmin}.

\section{Proof of Proposition~\ref{pro:gap}} \label{gap}

To prove Proposition~\ref{pro:gap}, we first derive a lower bound $\delta'_{lb}(\mu,r)$, which is looser than the lower bound $\delta_{lb}(\mu,r)$ in Proposition 2 but more tractable. The bound leverages Proposition~\ref{pro:lower}, Proposition~\ref{pro3} and the convexity of the minimum NDT $\delta^*(\mu,r)$ as stated in Lemma 1. The lower bounds $\delta_{lb}(\mu,r)$ and $\delta'_{lb}(\mu,r)$ are illustrated in Fig.~\ref{fig:minb}.

\begin{lemma} \label{lem:lbp}
For any $r \in [0,1]$, and $\mu$ with $\mu K_T\leq m_{max}$, we have $\delta'_{lb}(\mu,r)\leq \delta_{lb}(\mu,r)$, where $\delta_{lb}(\mu,r)$ is given in \eqref{eq:pro2} and we have defined  
\begin{align}  \label{eq:delpub}
\delta'_{lb}(\mu,r)= \frac{(i+2-\mu K_T) K_R}{(i+1)n_T}+  \frac{(\mu K_T-i-1)K_R}{(i+2)n_T} 
\end{align}
for $\mu K_T\in [i,i+1)$, with $m(r) \leq i\leq m_{max}-1$; and
\begin{subequations} \label{lem:dp}
  \begin{align}[left ={\delta'_{lb}(\mu, r)= \empheqlbrace}]
\mbox{\normalsize\(	\frac{K_R(m^*(r)-\mu K_T)}{K_T r}+\frac{K_R}{m^*(r) n_T}\)},  	  ~~~& \text{if}~ \mbox{\small \(\frac{K_T r}{n_T}\)}\in[ (m(r)-0.5)^2,m^2(r)] \label{eq:dp1} \\
\mbox{\normalsize\( \frac{K_R(m(r)-\mu K_T)}{K_T r} +\delta'_{lb}(\frac{m(r)}{K_T},r) \)},   ~& \text{if}~\mbox{\small \(\frac{K_T r}{n_T}\)}\in[m^2(r),(m(r)+0.5)^2] \label{eq:dp2}
\end{align}
\end{subequations}
for $\mu K_T \leq m(r)$, where $m^*(r)$ is given in \eqref{mmins}.
\end{lemma}

\begin{proof} 
From Proposition~\ref{pro3}, we have the equality $\delta^*(\mu,r)=\delta_{ach}(\mu,r)$ for $\mu K_T\in \{m(r)+1,\cdots,m_{max}\}$. Furthermore, we know that the minimum NDT $\delta^*(\mu,r)$ is a convex function of $\mu$ for any $r\geq 0$. Define $g_r(\mu)$ as a subgradient of the minimum NDT $\delta^*(\mu,r)$ at $\mu \in [0,1]$ for a fixed value of $r$. Consider any two points $\mu_{1}$ and $\mu_{2}$, where $\mu_{1} K_T\in \{m(r)+1,\cdots,m_{max}\}$ and $\mu_2$ is arbitrary. By a known convex property of convex functions (see~\cite{BD}), we have the inequality 
\begin{align}
\delta^*(\mu_2,r) \geq g_r(\mu_1)(\mu_2-\mu_1)K_T+ \delta^*(\mu_1,r). \label{subg}
\end{align} 
Therefore, choosing $\mu_{2}$ so that $\mu_{2}K_T=\mu_1 K_T+1$ in \eqref{subg} yields
\begin{align}
g_r(\mu_1)\leq \delta^*(\mu_1+1/K_T,r)-\delta^*(\mu_1,r). \label{dom:g}
\end{align} 

\begin{figure}[t!] 
  \centering
\includegraphics[width=0.4\columnwidth]{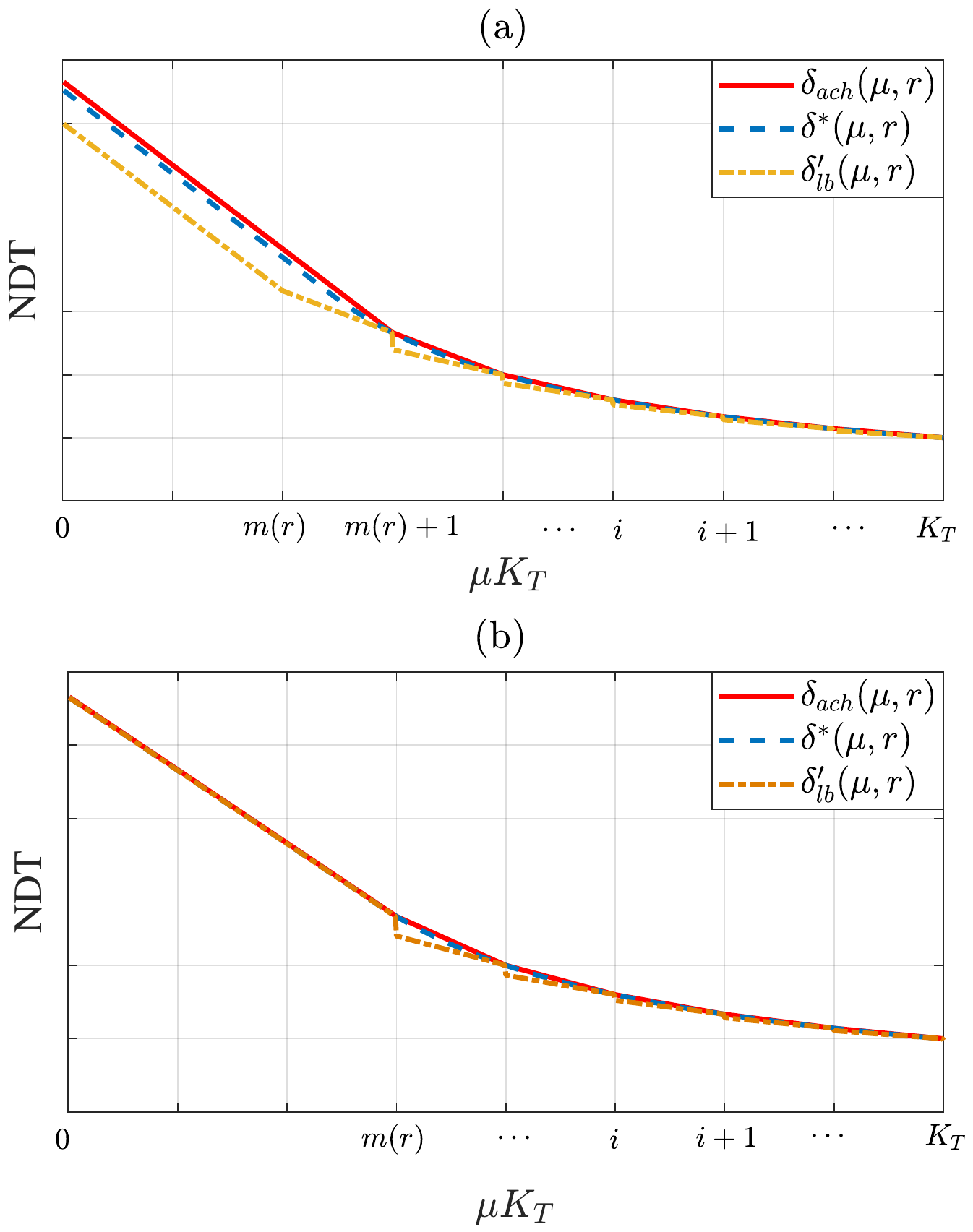}
\caption{Achievable NDT $\delta_{ach}(\mu,r)$ and lower bounds $\delta_{lb}(\mu,r)$ and $\delta'_{lb}(\mu,r)$: plot (a) shows Case 1 \eqref{eq:dp1}, plot (b) shows Case 2 \eqref{eq:dp2}.}
\label{fig:minb}
\end{figure}
For any sub-interval $\mu K_T \in [i, i+1)$, with $m(r) \leq i\leq m_{max}-1$, by setting $\mu_1=(i+1)/K_T$ in \eqref{dom:g}, we have the bound $g_r((i+1)/K_T)\leq g_{max}\defeq \delta^*((i+2)/K_T,r)-\delta^*((i+1)/K_T,r)$. Combining with \eqref{subg} and setting $\mu_2=\mu$, we have the inequality $\delta^*(\mu,r)\geq g_{max} \cdot(\mu K_T-i-1) + \delta^*((i+1)/K_T,r)$, which gives \eqref{eq:delpub} by Lemma 1. 

For the remaining interval $\mu K_T\leq m(r)$, we distinguish the two cases illustrated in Fig.~\ref{fig:minb}(a) and (b).

\emph{Case 1: $K_T r/n_T \in [(m(r)-0.5)^2, m(r)^2 ]$.} By \eqref{mmin} and \eqref{mmins} in this range, we have the inequality $m^*(r)=\sqrt{K_T r/n_T}\leq m(r)$. It can be directly verified that $\delta'_{lb}(\mu,r)$ in \eqref{eq:dp1} is no larger than $\delta_{lb}(\mu,r)$ in \eqref{eq:pro21} for $\mu K_T\leq m^*(r)$. Instead, for $\mu K_T \in[m^*(r), m(r)]$, 
since both $\delta'_{lb}(\mu,r)$ in \eqref{eq:dp1} and $K_R/(\mu K_T n_T)$ are decreasing functions of $\mu$, they are equal for $\mu K_T=m^*(r)$, and the former has a smaller gradient for the whole range of value of $\mu$ at hand, we have $\delta'_{lb}(\mu,r)\leq K_R/(\mu K_T n_T)$, which implies that $\delta'_{lb}(\mu,r)\leq \delta^*(\mu,r)$ in \eqref{eq:pro22}, as illustrated in Fig.~\ref{fig:minb}(a). 

\emph{Case 2: $K_T r/n_T \in [m(r)^2,(m(r)+0.5)^2 ]$.} By \eqref{mmin} and \eqref{mmins} in this range, we have the inequality $m^*(r)=\sqrt{K_T r/n_T}\geq m(r)$. By setting $\mu_1=(m(r)+1)/K_T$ in \eqref{dom:g}, we have $g_r((m(r)+1)/K_T)\leq g'_{max}\defeq \delta^*((m(r)+2)/K_T,r)-\delta^*((m(r)+1)/K_T,r)$. Combining with \eqref{subg} and setting $\mu_2=m(r)/K_T$, we have the inequality $\delta^*(m(r)/K_T,r)\geq -g'_{max}  + \delta^*((m(r)+1)/K_T,r)$, which gives the lower bound $\delta'_{lb}(m(r)/K_T,r)$. It is easy to verify the inequality $\delta'_{lb}(m(r)/K_T,r) \leq \delta_{lb}(m(r)/K_T,r)$. Combining this with the fact that $\delta'_{lb}(\mu,r)$ in \eqref{eq:dp2} and $\delta_{lb}(\mu,r)$ in \eqref{eq:pro21} are linear and parallel for $\mu K_T\leq m(r)$, we have $\delta'_{lb}(m(r)/K_T,r) \leq \delta_{lb}(\mu,K_T)$ in this range (see Fig.~\ref{fig:minb}(b)). This completes the proof.
\end{proof}

Using the lower bound $\delta'_{lb}(\mu,r)$, we can now directly compute the gap between the achievable NDT $\delta_{ach}(\mu,r)$ in Proposition 1 and the minimum NDT $\delta^*(\mu,r)$. Specifically, for $\mu K_T \in [i, i+1)$, with $m(r) \leq i\leq m_{max}-1$, from \eqref{ach:smallr} and \eqref{eq:delpub}, we verify that   
\begin{align}
\frac{\delta_{ach}(\mu,r)}{\delta'_{lb}(\mu,r)} \stackrel{(a)}{\leq} \frac{\delta_{ach}(\mu=i/K_T,r)}{\delta'_{lb}(\mu=i/K_T,r)} = 1+\frac{2}{i+3i}\leq \frac{3}{2},
\end{align}
where inequality (a) holds because $\delta_{ach}(\mu,r)$ and $\delta'_{lb}(\mu,r)$ are both linearly decreasing and they coincide at the endpoint $\mu K_T=i+1$. For $\mu K_T\leq m(r)$ in Case 1, from \eqref{ach:smallr} and  \eqref{eq:dp1}, the gap is given as 
\begin{subequations} 
\begin{align}
\frac{\delta_{ach}(\mu,r)}{\delta'_{lb}(\mu,r)} &\stackrel{(a)}{\leq} \frac{\delta_{ach}(\mu=m(r)/K_T,r)}{\delta'_{lb}(\mu=m(r)/K_T,r)} \label{final00} \\
                                                 &=\frac{1/m(r)n_T}{(m^*(r)-m(r))/(K_T r)+1/(m^*(r)n_T)} \\
																								& \stackrel{(b)}{\leq} \frac{m(r)}{\sqrt{m(r)(m(r)-1)}} \stackrel{(c)}{\leq} \sqrt{2},
\end{align}
\end{subequations}
where inequality (a) holds because $\delta_{ach}(\mu,r)$ and $\delta'_{lb}(\mu,r)$ decrease with the same slope and the maximum ratio is at the endpoint $\mu K_T=m(r)$; inequality (b) holds due to the constraints $m^*(r) \in [\sqrt{ m(r)( m(r)-1)}, m(r)]$; and inequality (c) holds for any $m(r)\geq 2$, while for $m(r)=1$, we have $K_T r/n_T \in[0,1]$ and $\mu\in[0,1]$, it has been proved that $\delta_{ach}(\mu,r)$ is optimal in Proposition \ref{pro3}. Finally, for $\mu K_T\leq m(r)$ in Case 2, from \eqref{ach:smallr} and  \eqref{eq:dp2}, the gap is given as 
\begin{align}
\frac{\delta_{ach}(\mu,r)}{\delta'_{lb}(\mu,r)}\stackrel{(a)}{\leq} \frac{\delta_{ach}(\mu=m(r)/K_T,r)}{\delta'_{lb}(\mu=m(r)/K_T,r)} = 1+\frac{2}{m^2(r)+3m(r)}\leq \frac{3}{2},
\end{align}
where inequality (a) holds as inequality (a) in \eqref{final00}, completing the proof.

\section{Proof of Proposition~\ref{pro:NDTpmin}} \label{sec:proofboundp}

Any achievable pipelined policy $\{\mathcal{C}_i, \{\mathcal{F}_i(b)\}_{b\in[B]}, \{\vv_{inf}(b)\}_{n\in[N], f\in[F],b\in[B]}\}_{i=1}^{K_T}$ can be converted into a serial policy with parameters $\{\mathcal{C}_i, \mathcal{F}_i, \{\vv_{inf}(b)\}_{n\in[N], f\in[F],b\in[B]}\}_{i=1}^{K_T}$, where $\mathcal{F}_i=\{\mathcal{F}_i(b)\}_{b\in[B]}$. In words, in the serial policy, all fronthaul transmission takes place prior to edge communications. Using the definitions in \eqref{def:TFP}-\eqref{def:TEP}, the fronthaul and edge latencies of the serial policy are given by $T_F= \sum_{b=1}^{B}T_F(b)$ and $T_E=\sum_{b=1}^{B}T_E(b)$. Furthermore, by the definition \eqref{pt}, we have the inequalities
\begin{align} \label{bound:TP}
T_P =\sum_{b=1}^{B} \max\{T_E(b), T_F(b)\} \geq \max\Big\{\sum_{b=1}^{B}T_E(b),\sum_{b=1}^{B}T_F(b) \Big\}\geq \max\{T_E, T_F\}. 
\end{align}
Finally, from \eqref{bound:TP}, we have the inequality 
\begin{align}
\delta_P\geq \max\{\delta_E, \delta_F\}, \label{pro:p}
\end{align} 
where $\delta_F$ and $\delta_E$ are the fronthaul and edge NDTs of the discussed serial policy. As a summary, any achievable pipelined NDT is lowered bounded by the maximum of the fronthaul and edge NDTs of the converted serial policy. 

Recall that, in Appendix~\ref{lemma2}, the fronthaul and edge NDTs under any serial policy are found to be lower bounded as \eqref{low:df} and \eqref{low:de}, respectively, which yields the inequalities
\begin{align} \label{low:p}
\delta_{E} \geq \frac{K_R}{n_T x}~\text{and}~\delta_{F} \geq \frac{K_R(x-\mu K_T)}{K_T r},
\end{align}
where $x=\sum_{i=1}^{K_T} i b_i/(NF)$ takes values in the interval $x\in [x_{min}, x_{max}]$ with $x_{min}=\max\{1,\mu K_T\}$ and $x_{max}=\max\{m_{max},\mu K_T\}$. To proceed, we define the two functions 
\begin{align} \label{fs}
f_1(x)=\frac{K_R}{n_T x}~\text{and}~f_2(x)=\frac{K_R(x-\mu K_T)}{K_T r}.
\end{align}
As a result, from \eqref{pro:p}, \eqref{low:p} and \eqref{fs}, the minimum pipelined NDT $\delta_P^*$ can be bounded as 
\begin{align}
\delta^*_{P} \geq \max\{f_1(x),f_2(x)\}. \label{low:p1}
\end{align}

To complete the proof, we now minimize the function $f(x)=\max\{f_1(x),f_2(x)\}$ in the interval $x\in [x_{min}, x_{max}]$. Function $f(x)$ is convex for $x > 0$, and the only point whose subdifferential $\partial f(x)$ includes 0 is $x=m_p^*$, i.e., $0 \in \partial f(m_p^*)$. Hence, from~\cite{BD}, the minimum value $f_{p,min} $ of $f(x)$ is given as
\begin{align} \label{fminp}
f_{p,min}= \left\{
\begin{array}{ll} 
f(m_p^*), &  \text{if}~ x_{min}\leq m_p^* \leq x_{max}\\
       \min\{f(x_{min}), f(x_{max})\}, & \text{otherwise},
\end{array} 
\right.
\end{align}
which equals to $\delta_{p,lb}(\mu,r)$ in \eqref{deltamin}.

\section{Proof of Proposition~\ref{pro:gapp}} \label{sec:proofgapp}
We now bound the gap between the achievable pipelined NDT $\delta_{p,ach}(\mu,r)$ in \eqref{achp} and the minimum NDT $\delta_{p}^*(\mu,r)$ by using the lower bound $\delta_{p,lb}(\mu,r)$ in \eqref{deltamin}. There are two cases in terms of $\mu K_T$. When $\mu K_T \leq 1-K_T r/n_T$, we have the equality $\delta_{p,ach}(\mu,r)=\delta_{p,lb}(\mu,r)$ by comparison, indicating that the achievable NDT is optimal and the gap is 1. 

We move to the case $\mu K_T \geq 1-K_T r/n_T$, which corresponds to $m_p^*\geq 1$ in \eqref{def:mps}. Directly from \eqref{mp} and \eqref{mpins}, we can obtain the multiplicities $m_p(\mu,r)=\min\{\lfloor m_p^*\rfloor, m_{max}\}$ and $m_p^*(\mu,r)=\min\{m_p^*,m_{max}\}$, respectively. By comparison, we have the inequality $m_p(\mu,r)\leq m_p^*(\mu,r)$. As a result, the gap can be bounded as
\begin{align}
\frac{\delta_{p,ach}(\mu,r)}{\delta_{p,lb}(\mu,r)} \stackrel{(a)}{\leq}\frac{\max\big\{ \frac{K_R}{m_p(\mu,r)n_T}, 1 \}}{\max\big\{ \frac{K_R}{m_p^*(\mu,r)n_T}, 1 \}} \stackrel{(b)}{\leq} \frac{m_p^*(\mu,r)}{m_p(\mu,r)} \stackrel{(c)}{\leq} \frac{m_p^*}{\lfloor m_p^* \rfloor} \stackrel{(d)}{\leq} 2,
\end{align}
where inequality (a) holds because $\delta_{p,ach}(\mu,r)$ is a lower convex envelope of $\max\big\{ K_R/(m_p(\mu,r)n_T), 1 \}$; inequality $(b)$ holds by considering the three different cases: $K_R/m_p(\mu,r)n_T \leq 1$, $K_R/m_p^*(\mu,r)n_T \geq 1$, and $K_R/m_p^*(\mu,r)n_T \leq 1 \leq K_R/m_p(\mu,r)n_T$, respectively, along with the fact that $K_R/(m_p(\mu,r)n_T)\geq  K_R/(m_p^*(\mu,r)n_T)$; inequality (c) holds by considering the three cases: $m_p^* \leq m_{max}$, $\lfloor m_p^*\rfloor\geq m_{max}$, and $\lfloor m_p^*\rfloor \leq m_{max} \leq m_p^*$; and inequality (d) holds because $m_p^* \geq 1$. This completes the proof. 
\end{appendices}

\section*{Acknowledgements}
Jingjing Zhang and Osvaldo Simeone have received funding from the European Research Council (ERC) under the European Union's Horizon 2020 Research and Innovation Programme (Grant Agreement No. 725731). The authors would like to thank Roy Karasik for useful comments. 

\bibliographystyle{IEEEtran}
\bibliography{IEEEabrv,final_refs}

\end{document}